\documentclass[reqno]{amsart}

\usepackage{graphics}
\usepackage{color}

\usepackage{
 amsfonts
,amssymb
,amsmath
,mathabx
,colonequals
,amsthm
,yfonts
,fge
}%

\usepackage{a4wide}

\newcommand*{\op}{{\mathcal L}}

\newcommand*{\mtr}{{\mathrm L}}

\newcommand*{\opmA}{\op_{A}}
\newcommand*{\opmB}{\op_{B}}
\newcommand*{\opmC}{\op_{C}}

\newcommand*{\matrA}{\mtr_{A}}
\newcommand*{\matrB}{\mtr_{B}}

\newcommand*{\ppp}{{\mathfrak{p}}}
\newcommand*{\qqq}{{\mathfrak{q}}}
\newcommand*{\rrr}{{\mathfrak{r}}}
\newcommand*{\sss}{{\mathfrak{s}}}

\newcommand*{\pp}{{{p}}}
\newcommand*{\qq}{{{q}}}
\newcommand*{\rr}{{{r}}}
\newcommand*{\sS}{{{s}}}

\newcommand*{\famA}{{\mathrm A}}
\newcommand*{\famB}{{\mathrm B}}
\newcommand*{\famC}{{\mathrm C}}

\newcommand*{\mllll}{{\ell}}

\newcommand*{\delTa}{{\mathcal D}}
\newcommand*{\Id}{{\mathbb I}}

\newcommand*{\argz}{}

\newcommand{\half}{\mbox{\scriptsize$1\over2$}}

\newcommand*{\Eqs}{{Eq.s}}

\newcommand{\Imi}{\mathrm{i}}

\newcommand{\eIphi}{{\Phi}}

\newcommand{\eP}{{\Psi}}

\newcommand{\somE}{    E    }
\newcommand{\EEpm}[1]{E_{\{\! #1 \!\}}}

\newcommand{\covered}[1]{\mbox{\textswab{#1}}}

\newcommand{\coveredz}{\raisebox{+0.012em}{\covered{z}}}

\newcommand{\coveredi}{
\raisebox{+0.012em}[1.2ex][0px]{\covered{i}}
}

\newcommand{\coveredmi}{
\raisebox{+0.012em}[1.2ex][0px]{-\covered{i}}
}
\newcommand{\coveredone}{
\raisebox{-0.050em}{\covered{1}}
}

\newcommand{\coveredmcone}{\raisebox{0px}[1em][0px]{\coveredco{-$\coveredone$}}}

\newcommand{\coveredmpone}{\raisebox{0px}[1em][0px]{\!\!\!\coveredpro{-$\coveredone$}}}

\newcommand{\coveredmiv}{\raisebox{0px}[1em][0px]{\coveredco{\,-$1/$}}}

\newcommand{\coveredrevco}{\raisebox{0px}[1em][0px]{\coveredco{-}}}
\newcommand{\coveredrevpro}{\raisebox{0px}[1em][0px]{\!\!\!\coveredpro{-}}}

\newcommand{\diag}{\,{\mathrm{diag}}}

\newcommand{\Mono}{\mathcal{M}}

\newcommand{\puncturedS}{^\backprime{\!}S^1}

\newcommand{\coveredpro}[1]{
\mbox{\hspace{-0.2em}%
\raisebox{0.5em}{
${\mbox{\scalebox{0.8}[0.6]{\rotatebox{180}{$\curvearrowright$}}}
\atop
 \raisebox{-0.13em}{#1}
}$%
}\hspace{-0.4em}
}}

\newcommand{\coveredco}[1]{\mbox{\hspace{-0.4em}%
\raisebox{0.5em}{%
${
\raisebox{-0.25em}[1ex][0.1em]{
\mbox{\scalebox{0.8}[0.5]{$\curvearrowleft$}}
}
\atop
 \raisebox{-0.13em}{#1}
}
  $
}\hspace{-1.02em}
}}

\newcommand{\coveredcominus}{\raisebox{0px}[0.81em][0.px]{\coveredco{$-$}}}
\newcommand{\coveredprominus}{\raisebox{0px}[0.81em][0.px]{\coveredpro{$-$}}}

\newtheorem{theorem}{Theorem}
\newtheorem{corollary}{Corollary}
\newtheorem{lemma}{Lemma}
\newtheorem{proposition}[theorem]{Proposition}

\theoremstyle{definition}

\newtheorem{convention}{Convention}

\newcommand{\involutive}{involutive}  
\newcommand{\monodromy}{monodromy}  
\newcommand{\subdomain}{subdomain}  

\begin{document}

\title[Symmetries of sDCHE]%
{
Symmetries of the space of solutions to
special double confluent Heun equation 
of negative integer order and its applications
}

\author{Sergey I. Tertychniy
}%

\address%
 {
 Russian Metrological Institute of Technical Physics and Radio Engineering
 (VNIIFTRI), Mendeleevo, 141570, Russia}%
 \thanks{
 Supported in part by RFBR grant N 17-01-00192.
 }
\begin{abstract}\noindent

Three linear operators ($\op$-operators) determining automorphisms 
of the space of solutions to a special double
confluent Heun equation (sDCHE) {\it  of negative integer order \/}
are considered. 
Their composition rules involving in a natural way the \monodromy{} transformation are given.
Introducing
eigenfunctions $\EEpm{+}, \EEpm{-}$ of one of $\op$-operators ($\opmC$)
which also satisfy sDCHE,
the four polylocal quadratic
functionals playing role of the first integrals of sDCHE
are derived. 
Basing on them,
the explicit matrix representations
of $\op$-operators
and the \monodromy{} operator with respect to the basis constituted by $\EEpm{\pm}$
are constructed.
The composition rules of $\op$-operators 
lead to functional equations for the eigenfunctions $\EEpm{\pm}$
which
can be interpreted as analytic continuations
of solutions to sDCHE from the half-plane $\Re z>0$ to their whole domain
by means of 
algebraic transformations.
Application of the  above results to the theory of the first order 
non-linear differential equation utilized, in particular, 
for the modeling of overdamped Josephson junctions in superconductors and
closely related to sDCHE
is presented.
The automorphisms of the set of its solutions induced
by $\op$-operators and by the \monodromy{} operator
(certain shift of the solution domain
in the latter case) represented by algebraic operations
are found.
Among them, one transformation is \involutive{} 
and another can be regarded as the square root of
the transformation induced by the \monodromy{} transformation.
\end{abstract}

\maketitle

Let us consider the following linear homogeneous second order differential equation
\begin{equation}
\label{eq::0010}
z^2 \somE''
+\big((l+1)
 z
+ \mu (1-z^2) \big) \somE'
+
(  -\mu  (l+1) z
+
\lambda
%
)
\somE
=0,
\end{equation}
where $l,\mu\not=0,$ and $\lambda$ are some constant parameters,
the argument 
of the holomorphic 
function $\somE=\somE(\coveredz)$ 
is 
a point $\coveredz$
of the universal cover $\covered{C}^*$
of $\mathbb{C}^*=\mathbb{C}\fgebackslash 0$,
the complex plane with zero removed,
$z$ denoting result of the canonical projection 
$\iota: \covered{C}\to \mathbb{C}^*$
of
$\coveredz$, 
$z=\iota\,\coveredz$. 
The Riemann surface
$\covered{C}^*$ 
plays the role of 
the natural domain of holomorphic
solutions to Eq.~\eqref{eq::0010}  singular at zero.
Its use 
is necessary for the capturing of some their global properties but
most of local considerations, including establishing of fulfillment of
Eq.~\eqref{eq::0010}, allows to replace the points of $\covered{C}^*$ by their projections
(complex numbers)
living in
$ \mathbb{C}^* $.

Eq.~\eqref{eq::0010} belongs to the family of
so called
double confluent Heun equations \cite{SW,SL} which name is often abbreviated to DCHE.
In Ref.~\cite{SW} a generic DCHE is given in the following form
\begin{equation}
                   \label{eq::0020}
D^2 v +\alpha(z+z^{-1}) D v + \big((\beta_1+1/2)\alpha z+\alpha^2/2-\gamma+(\beta_{-1}-1/2)\alpha z^{-1}\big)v=0,
\end{equation}
where $v=v(z)$ is the unknown, $D=z\,d/d\,z$, and  $\alpha,\beta_1, \beta_{-1}, \gamma$ are
arbitrary complex constants.
A straightforward computation shows that the ansatz 
\begin{eqnarray}
&&v=z^{l/2} e^{-\mu z}E,\nonumber
\\
&&\alpha=\mu,\,\beta_1=-(l+1)/2,\,\beta_{-1}=(1-l)/2,\,
\gamma=l^2/4-\lambda-\mu^2/2\nonumber
\end{eqnarray}
leads just to Eq.~\eqref{eq::0010}
(up to local interpretation of the argument of unknown). 
Eq.~\eqref{eq::0010}  is therefore equivalent to Eq.~\eqref{eq::0020},
provided the constraint
\begin{equation}
\label{eq::0050}
\beta_1+1/2= \beta_{-1}-1/2
\end{equation}
is imposed. Since Eq.~\eqref{eq::0010}  is characterized by the three free constant parameters while
there are four such parameters in Eq.\ \eqref{eq::0020}, Eq.~\eqref{eq::0010}  is sometimes named {\it special\/}
double confluent Heun equation.
The reason why equations of the form \eqref{eq::0010}    
were segregated within the general DCHE family will be discussed below.

In the present notes, we consider a particular case of Eq.~\eqref{eq::0010}  which arises when the parameter
$l$ (sometimes named its order) is integer. Moreover, since the case of {\it positive\/} integer
$l$ has been studied elsewhere (see \cite{T1}), we  claim also of $l$ to be negative, i.e.\ assume that
\begin{equation}
\label{eq::0060}
  l=-\mllll \;\mbox{ for some  } \mllll\in\mathbb{N}.
\end{equation}
A particular 
interest to Eq.~\eqref{eq::0010} of integer order is motivated by discovering
of additional symmetries of the space $\mathbf\Omega$ of its solutions. It is manifested
by existence of the three linear operators  
admitting 
simple explicit representations
which take $\mathbf\Omega$ onto itself \cite{BT1}.
Moreover, apart of rather special values of parameters, these operators
are invertible and thus act on $\mathbf\Omega$ as automorphisms.
In conjunction with the \monodromy{} map,
these give rise to
the group of symmetries of $\mathbf\Omega$.

Specifically, let us denote by the symbols $\opmA, \opmB, \opmC  $
the transformations of functions $\somE$ holomorphic in $\covered{C}^*$ which are described by the following
formulas:
\begin{eqnarray}
   \label{eq::0070}
\opmA:\somE(\coveredz)\mapsto
\opmA[{\somE}](\coveredz)
&=&
z^{2(\mllll-1)}e^{\mu(z+{1/z})}
\hspace{-0.1em}
\raisebox{-3pt}[1pt][5pt]{$\biggl\lfloor$}%
_{
\mbox{
\rlap{\hspace{-0.3ex}$
\phantom{.g|
}
\atop
\coveredz
\leftleftharpoons
-{{1\!/\coveredz}}^{\vphantom{\mathstrut}},\,
z
\leftleftharpoons
-{1/z}
$
}
}                                     
}\hspace{-2.53ex}\
\big[
-\ppp\big({\somE}'-\mu {\somE}\big)
+(-1)^\mllll \rrr {\somE}
\big],
\\
%
    \label{eq::0080}
\opmB:\somE(\coveredz)\mapsto
\opmB[{\somE}](\coveredz)
&=&
2\omega\,
z^{1-\mllll}e^{\mu(z+{1/z})}
 \times
\nonumber\\
&&\hspace{-0.20em}
\times
\raisebox{-3pt}[1pt][5pt]{$\biggl\lfloor$}%
_{
\mbox{
\rlap{\hspace{-0.6ex}$
\phantom{.b|
}
\atop
\coveredz
\leftleftharpoons
-{{\coveredz}^{\vphantom{\mathstrut}}}^{\vphantom{\mathstrut}}
,\,
z
\leftleftharpoons
-z
$
}
}
}\hspace{-2.53ex}\
\big[
(\qqq+\mu z^2\ppp) \big({\somE}'-\mu {\somE}\big)
-(-1)^\mllll(\sss+\mu z^2\rrr){\somE}
\big]
,
\\
%
\label{eq::0090}
\opmC:\somE(\coveredz)\mapsto
\opmC[\somE](\coveredz)
&=&
2\omega\,
z^{\mllll-1}
\raisebox{-3pt}[1pt][5pt]{$\Bigl\lfloor$}%
_{
\mbox{
\rlap{\hspace{-0.1ex}$
\hphantom{.
}
\atop
\coveredz
\leftleftharpoons
{{1\!/\coveredz}}
$
}
}
}
\hspace{-2.1ex}
\big( \somE'  - \mu \somE \big).
\end{eqnarray}
In them,
$\omega$ is an arbitrary non-zero constant while the symbols
$\rrr,\qqq,\rrr,\sss $ denote certain functions of $z$ 
(depending also on the parameters $\mllll,\lambda,\mu$)
which are defined
as follows.
\begin{enumerate}
\item
  Let us introduce the two sequences of the pairs of functions
$\{\pp_k, \qq_k\},\{\rr_k,\sS_k\}$, $k=0,1,2,\dots$,
of the variable  $z\in\mathbb{C}^*$
by means of
the following recurrence schemes:
\begin{eqnarray}
\label{eq::0100}
&&\begin{aligned}\hspace{0.9em}
\pp_{0}=0,\; \qq_{0}=1,\;\rr_{0}=z^{-2},\; \sS_{0}=-\mu;
\end{aligned}
\\
\label{eq::0110}
&&
\left\{
\begin{aligned}
\pp_{k}=&\:
(1 - \mllll)z\,\pp_{k-1} + \qq_{k-1} + z^2 \pp_{k-1}',
\hspace{5.9em} \mbox{for }k\in\mathbb{N}
 \\
\qq_{k}=&\:
z^2 (-\lambda + (\mllll+1) \mu z) \pp_{k-1} + \mu\left(1- z^2\right) \qq_{k-1}
+ z^2 \qq_{k-1}'
\\[0.2em]
\end{aligned}
\right.
\\
\label{eq::0120}
&&
\left\{
\begin{aligned}
\rr_{k}=&\:
2 (k-2)z\,\rr_{k-1} - \sS_{k-1} - z^2 \rr_{k-1}',
\\
\sS_{k}=&\:
z^2 \left(\lambda-\left(\mllll + 1\right) \mu z\right)\rr_{k-1}+
\\&
+ \left(\left(2(k-1)-(\mllll+1)\right) z + \mu\left(z^2-1\right) \right)\sS_{k-1}
- z^2 \sS_{k-1}';
\end{aligned}
\llap{for $k\in\mathbb{N}$}
\right.
\end{eqnarray}
\item 
let us pick out their ``diagonal'' elements 
for which 
the pair's
indices  are equal to $\mllll $
and assign
\begin{equation}
\label{eq::0130}
 \ppp= \pp_{\mllll},\;\qqq= \qq_{\mllll},\;\rrr= \rr_{\mllll},\;\sss=\sS_{\mllll}.
\end{equation}
\end{enumerate}
Albeit $\rr_{0}$ is singular at $z=0$,
for any greater index $k$ all the functions
$\pp_k, \qq_k,\rr_k,\sS_k$  are polynomial in $z$.
In particular,
the functions \eqref{eq::0130} are also polynomials and it can be shown
that their degrees are equal to
$2(\mllll-1), 2\mllll,2(\mllll-1), 2\mllll$, respectively \cite{BT1};
accordingly, the functions
$\ppp=\ppp(z),
\dots
\sss=\sss(z) $ 
are 
well defined and are holomorphic everywhere in  $\mathbb{C}$.
(It is worth noting that they are 
polynomial in $\lambda$ and $\mu$ as well although this dependence
is not manifested in our formulas).
Thus, in the formulas \eqref{eq::0070}, \eqref{eq::0080}, 
the symbols $\ppp,\!\qqq,\!\rrr,\!\sss$
stand for 
the known
 entire functions of $z=\iota\,\coveredz$
which can be computed
in explicit form
for every given $\mllll$.
As above, the symbols
  $\somE,\somE'$ also denote some holomorphic functions  but
their argument is 
$\coveredz\in\covered{C}^*$.
Finally, the  ``aggregates'' of the form
$\lfloor_%
{\,\coveredz
\leftleftharpoons
{{\covered{u}}}
\!,\,
z
\leftleftharpoons
\mathrm U}
\{\cdots\} 
$
denote  the results 
of substituting
of the expressions
$ {\covered{U}} $ and $\mathrm U $ in place of
the variables
$ \coveredz $ and $ z $, respectively, into the   expression
$ \{\cdots\}  $ considered as a function of 
$ \coveredz $ and $ z $.

Having explained 
the
notations, it is important to note that
\Eqs~\eqref{eq::0070}-\eqref{eq::0090} can not still serve
themselves exhaustive
definitions of certain transformation rules as they should.
Their insufficiency
lurks in 
the lack of uniqueness inherent to the lifts
of the maps
$\famA\,|\,\famB\,|\,\famC:z \to -{1/z}\,| -\!z\,|\, {1/z}$
from the punctured complex plane $\mathbb{C}^*$ to its
universal cover $\covered{C}^*$. In other words, one must
additionally
 clarify
what is the actual meaning of the records $-{1\!/\coveredz}, -\coveredz, {1\!/\coveredz}$
encompassed in 
\eqref{eq::0070}-\eqref{eq::0090},
the point which is simple but yet plays sometimes a non-trivial role. 

To be more specific, let us model 
the Riemann surface  $\covered{C}^*$
by the set $\mathbb{R}_+\!\times\mathbb{R}$
of all the pairs $(\rho,\phi)$ of real numbers (coordinates)
of which
the first one 
is strictly positive, $\rho>0$.
Then the projection $\iota$
of  $\covered{C}^*$
to $\mathbb{C}^*$
is defined by the map $ \covered{C}^*\ni\coveredz:=(\rho,\phi)\stackrel{\iota}{\to}\rho\, e^{\Imi\phi}=z=:
\iota\,\coveredz
\in \mathbb{C}^*$. We shall refer to the above construction as {\it the similog model\/} 
(of  $\covered{C}^*$).
It allows to identify 
simply connected subsets of $\mathbb{C}^*$ with the corresponding subsets
of $\covered{C}^*$. 
Then 
 functions holomorphic in $\covered{C}^*$
are defined
as those which
locally coincide with functions holomorphic in $\mathbb{C}^*$.
The algebraic operations
in $ \covered{C}^* $ must be
conveyed through the projection
$\iota$ to the corresponding operations in $\mathbb{C}^*$.

Let us consider, say, how the map $\famB: z\to-z$ can be lifted to $\covered{C}^*$.
In frames of the above model, 
the 
obvious
concordant transformations 
are
$\coveredcominus 
:(\rho, \phi)\to(\rho,\pi+\phi)$
(in projection to $\mathbb{C}^*$, 
the counterclockwise rotation through $\pi$ which is homotopic to the identical transformation)
 and
$ \coveredprominus 
:(\rho, \phi)\to(\rho,-\pi+\phi)$
(similar but this time clockwise rotation through the same angle, respectively).
These lifts of $\famB$ are distinct and
neither of them is preferable. 
The projections of the results of their actions to a point, at first, coincide
and, at second, differ from the projection of the argument by the opposite sign alone.
For example, ``reversing sign''
of the point $\coveredone:=(1,0)$ projected to 1, 
we 
obtain the
 two distinct 
{\it points} of $\covered{C}^*$,
$ \coveredmcone 
:=(1,\pi)$ 
and 
$ \coveredmpone 
:=(1,-\pi)$, which
both project to (in $\covered{C}^*$, ``play the role of'') {\it the number} $-1$.

For the inversion map $\famC:z\to 1/z$,
the eye-catching 
transformation which could 
serve its lift
to $\covered{C}^*$,
is obviously
$ 1/ : (\rho, \phi)\to(1/\rho,-\phi)$.
Although it 
may seem that nothing more is needed, 
it would actually be insufficient to limit our choice
by such an instance alone. Indeed, on $\mathbb{C}^*$,  the
map $C$ 
possesses the two fixed points, $z=+1$
and $z=-1$. At the same time,
the above lift
 `$ 1/ $' 
has the only fixed point
$\coveredone:=(1,0)$, which projects to $+1$, obviously.
Hence for the sake of completeness we have to introduce yet 
 another lift of the
inversion map which would possess a fixed point projected to $-1$. 
But we have noted above 
two points
with such projection, being not able to choose a preferable one among them.
The corresponding lifts of the inversion map are easily constructed:
the map
$ \raisebox{0px}[1em][0px]{\coveredco{1/}} : (\rho, \phi)\to(1/\rho,2\pi-\phi)$
possesses the only fixed point
$ \coveredmcone$,
the only fixed point of the map
$ \raisebox{0px}[1em][0px]{\coveredpro{1/}} : (\rho, \phi)\to(1/\rho,-2\pi-\phi)$
is $ \coveredmpone$.
Thus,
in total, the three candidates to the role of lift of the inversion map $C $ to $ \covered{C}^* $
have arisen. 

In case of the third \involutive{} map $\famA: z\to-1/z$ the use of the fixed points is also
a convenient
 trick helping to figure out the ``minimal collection'' of its lifts.
Now, in  $\mathbb{C}^*$, the fixed points are the pure imaginary $\Imi$ and $-\Imi$. The points projecting to them
which are most close both between themselves and also to the lift of the unit $\coveredone$
(an additional criterium for our selection)
are $\coveredi:=(1,\pi/2)$ and $\coveredmi:=(1,-\pi/2)$. They are the fixed points of the maps
  $\;\raisebox{0px}[1em][0px]{\coveredco{-$1/$}} : (\rho, \phi)\to(1/\rho,  \pi-\phi)$
and
$ \raisebox{0px}[1em][0px]{\coveredpro{-$1/$}} : (\rho, \phi)\to(1/\rho,-\pi-\phi)$, respectively,
which both
represent the
lifts of $A$.

It follows from the above examples that although the
replacements 
$z\leftleftharpoons -1/z \,|\, -\!z \,|\, 1/z$
utilized in formulas \eqref{eq::0070}-\eqref{eq::0090}, %
where they
act to arguments of holomorphic functions, can be extended to the case
of the domain $ \covered{C}^* $, the
triplet of the corresponding lifts is in no way unique.
To avoid such an uncertainty, 
we pick 
the only
 instance among the $2\times3\times2$ of
opportunities 
ensured by the above speculation (their amount is actually unlimited)
and assume to adhere to the following
\begin{convention}
Let us assume
the replacements of the argument
$\coveredz\in \covered{C}^* $ of the function $\somE$ and its
derivative $\somE'$ involved in the formulas \eqref{eq::0240}-\eqref{eq::0260}
to be realized
in accordance with 
the following rules
\begin{eqnarray}
\label{eq::0140}
\mathrm{A}:&& \coveredz
:= 
(\rho,\phi)
\leftleftharpoons 
(1/\rho,\pi-\phi)
=: 
\coveredco{\,-$1/$}\,\coveredz
,
\\[-0.3em]
\label{eq::0150}
\mathrm{B}:&&\coveredz
:= 
(\rho,\phi)
\leftleftharpoons 
(\rho,\pi+\phi)
\hspace{1.08em}
=:
\coveredco{-}\!\coveredz
,
\\
\label{eq::0160}
\mathrm{C}:&&\coveredz
:= 
 (\rho,\phi)
\leftleftharpoons 
(1/\rho,-\phi)
\hspace{1.08em}
=: 
1/\coveredz
,
\end{eqnarray}
respectively.
\end{convention}

It is worth noting that the different realizations of lifts
of the same transformations 
arise
from ones indicated in \Eqs{~}\eqref{eq::0140}-\eqref{eq::0160} 
by means of subtraction/addition of 
$2\pi$
from/to the second ($\phi$-) component of the ``conventional'' result.
In other words,
they can be obtained applying additionally a number of 
the replacements
\begin{eqnarray}
\label{eq::0170}
\mathrm{M}^{\pm1}:&& \coveredz
:=
(\rho,\phi)
\leftleftharpoons 
(\rho,\pm2\pi+\phi)
=: 
\mathrm{M}^{\pm1}\coveredz
\end{eqnarray}
which represent the {\it \monodromy{}} transformation and its inverse.
Application of arbitrarily folded \monodromy{}
transformation $\mathrm{M}^{k}, k\in\mathbb{Z},$ sends any lift of a map on $\mathbb{C}^*$
to another its lift.
There exist therefore not two or three lifts,
as it was considered in the above examples,
but rather two-side sequences of them.

Basing on definitions \eqref{eq::0140}-\eqref{eq::0160}, 
we are now able to formulate the main property of the $\op$-operators 
\cite{BT1}
which will be utilized here.
                       \begin{theorem}\label{t:010}
                    Transformations \eqref{eq::0070}-\eqref{eq::0090}
                     send any solution to Eq.~\eqref{eq::0010}
                     to solution of the same equation.
                        \end{theorem}\noindent
In other words,
these leave the space $\mathbf\Omega$
of solutions to Eq.~\eqref{eq::0010} invariant.
\begin{proof}[Proof outline]
Let us substitute the  right-hand side expressions
of the formulas
\eqref{eq::0070}-\eqref{eq::0090}
to Eq.~\eqref{eq::0010} and eliminate from the result
the higher derivatives $\somE'',\somE''' $ with the help of
the same equation \eqref{eq::0010}.
In cases of $\opmA$ and $\opmB$, the resulting expressions
contain also the first and second order derivatives of the polynomials 
$\ppp,\!\qqq,\!\rrr,\!\sss$. They are eliminated with the help of the equations
\begin{equation}
\label{eq::0180}
\begin{aligned}
z^2 \ppp'\argz =\;&\big(\mu + (\mllll-1) z\big) \ppp\argz - \qqq\argz + (-1)^\mllll z^2 \rrr\argz
,
\\
\qqq'\argz =\;& \big(\lambda - (\mllll+1) \mu z\big) \ppp\argz + \mu\, \qqq\argz + (-1)^\mllll \sss\argz
,
\\
z^2 \rrr'\argz =\;&(-1)^{\mllll+1} \big(\lambda + \mu^2\big) \ppp\argz + z \big(2 (\mllll-1) - \mu z\big) \rrr\argz - \sss\argz
,
\\
z^2 \sss'\argz =\;&(-1)^{\mllll+1} \big(\lambda + \mu^2\big) \qqq\argz
+ z^2 \big(\lambda - (\mllll+1) \mu z\big) \rrr\argz
+ \big((\mllll-1) z-\mu\big) \sss\argz
.
\end{aligned}
\end{equation}
which these functions obey (see \cite{BT2}). 
Then,
in all the three cases,
algebraic
simplifications finish with 
identically
zero result.
\end{proof}

Leaning on the specific properties of the functions $\ppp,\!\qqq,\!\rrr,\!\sss$,
in every other way,
the proof reduces to straightforward computations due to explicit
form of the basic formulas
\eqref{eq::0070}-\eqref{eq::0090}. Similar (though somewhat more lengthy)
computations allows one to obtain the {\it composition rules\/}
of the $\op$-operators, provided their domains
are restricted to the space $\mathbf\Omega$
of solutions to Eq.\ \eqref{eq::0010}. Namely, these read
\begin{eqnarray}
\label{eq::0190}
\opmA\circ\opmA&=&-\delTa\cdot \Id 
,\\
\label{eq::0200}
\opmB\circ\opmB&=&
(2\omega)^2(\lambda+\mu^2)
\delTa\cdot \Mono
,\\
\label{eq::0210}
\opmC\circ\opmC&=&
(2\omega)^2(\lambda+\mu^2)\cdot
\hphantom{-} \Id, 
\\
\label{eq::0220}
\opmA\circ\opmB&=&
  \hphantom{-}\delTa\cdot\Mono^{-1}\circ\opmC
,\\
\label{eq::0230}
\opmB\circ\opmA&=&
-
 \delTa
\cdot \opmC
,\\
\label{eq::0240}
\opmB\circ\opmC&=&
  -
(2\omega)^2(\lambda+\mu^2)\cdot
\Mono\circ\opmA
,\\
\label{eq::0250}
\opmC\circ\opmB&=&\hphantom{-}
(2\omega)^2(\lambda+\mu^2)\cdot
\opmA 
,\\
\label{eq::0260}
\opmC\circ\opmA&=&\hphantom{-}\opmB 
,\\
\label{eq::0270}
\opmA\circ\opmC&=&
  -\Mono^{-1}\circ\opmB
\\
  \label{eq::0280}
&&\hspace{-6.5em}
\rlap{where $\delTa=z^{2(1-\mllll)}\big(\ppp\sss-\qqq\rrr\big)$.}
\end{eqnarray}
Here $ \Mono $ and $ \Mono^{-1} $ denote the \monodromy{} transformation 
and its inverse which act on holomorphic functions in $\covered{C}^*$.
We have
defined the \monodromy{} transformation $\mathrm M$ in action to points of
$\covered{C}^*$ by the rule \eqref{eq::0170}.
For functions on $\covered{C}^*$, the \monodromy{} transformation
is inferred from it 
as an ancillary one as follows
$$%
(\Mono^{\pm1}\somE)(\coveredz)=\somE(\mathrm{M}^{\pm1}\coveredz{}).
$$
In a more thorough language
transformation $\Mono\equiv\Mono^1 $ 
can also be described 
as the result of
point-wise analytic continuations of $\somE$ 
along the arcs projected to
full circles in ${\mathbb{C}}^*$
with centers situated at zero
which are passed in
the
anticlockwise direction (clockwise for inverse \monodromy{} transformation $\Mono^{-1}$).
On lifts of functions holomorphic on ${\mathbb{C}}^*$ the transformations
$\Mono^{\pm1}$ act as the identical map.
Eq.\ \eqref{eq::0010}
is invariant with respect to such operations (which are the linear operators in fact)
as well as the
space of its solutions. More generally, the set of functions
holomorphic in  $\covered{C}^*$ 
is \monodromy{} invariant.

The \monodromy{} transformation arises in formulas \eqref{eq::0190}-\eqref{eq::0270}
as a consequence of composition rules for transformations of arguments
of the functions $\somE$ and $\somE'$ involved in
the formulas
\eqref{eq::0070}-\eqref{eq::0090}.
The results of
such pairwise compositions
are described by the following table of rules
\begin{equation}\label{eq::0300} 
\mbox{
\begin{tabular}[c]{l|ccc }
 $\circ$ & $ \mathrm{A}$& $ \mathrm{B} $  & $\mathrm{C} $   \\[0.1ex]
\hline\\[-1.9ex]
$ \mathrm{A} $ & $ Id  $ & $ \mathrm{M}^{-1}\circ \mathrm{C}  $ & $ \mathrm{M}^{-1}\circ \mathrm{B} $   \\
$ \mathrm{B} $ & $ \mathrm{C}$         & $ \mathrm{M} $  & $ \mathrm{M} \circ \mathrm{A} $   \\
$ \mathrm{C} $ & $  \mathrm{B}$        & $  \mathrm{A} $ & $  Id  $   \\
\end{tabular}}
\end{equation}
which is a direct consequence of
definitions
\eqref{eq::0140}-\eqref{eq::0160}.

In view of appearance of the \monodromy{} operator, the collection of the
composition rules \eqref{eq::0190}-\eqref{eq::0270},
binding pairwise all the $\op$-operators,
 proves to be
non-closed or, better to say, incomplete. Indeed, once 
the \monodromy{} transformation had  there appeared,
one must add to them
the rules describing compositions of $\op$-operators with the very \monodromy{} operator.
These take the form of commutation rules and read
\begin{eqnarray}
\label{eq::0310}
\Mono^{\pm1}\circ
\opmA
&= &
\opmA\circ\Mono^{\mp{}k},
\;
\\
\label{eq::0320}
  \Mono^{\pm{}k}\circ
\opmB
&= &
\opmB\circ \Mono^{\pm{}k},
\\
\label{eq::0330}
  \Mono^{\pm{}k}\circ
\opmC
&= &
\opmC\circ \Mono^{\mp{}k},\;\forall\, k\in\mathbb{Z}.
\end{eqnarray}
Thus, in the first and last cases, the commutation inverses
the \monodromy{} operator but keeps it unaltered in the second case.

The two comments concerning the {\it proof\/} of
\Eqs\
\eqref{eq::0190}-\eqref{eq::0270} are now apt.

{\it At first}, as it has been noted, they can be obtained by means of
straightforward computation  on the base of the formulas
\eqref{eq::0070}-\eqref{eq::0090}. Since they are only valid on the space of
solutions to Eq.\ \eqref{eq::0010}, it is used in calculations
for elimination of higher derivatives of the function $\somE$. However,
just like as in the proof of the theorem \ref{t:010},
the crucial role play here the specialties of the functions
$\ppp,\!\qqq,\!\rrr,\!\sss$. 
Indeed, expanding the left-hand sides of \Eqs{~}\eqref{eq::0190}-\eqref{eq::0270}, their derivatives arise.
As above, these are eliminated with the help of
\Eqs{~}\eqref{eq::0180}.
Another effect to the functions $\ppp,\!\qqq,\!\rrr,\!\sss$ is produced by the
replacements of their argument $z$. As a result, the expressions with some
$\ppp,\!\qqq,\!\rrr,\!\sss$ with arguments either $-1/z$ or $-z$ or $1/z$
arise. Uniformity of arguments is achieved
by means of the following transformations\\[0.5em]
\noindent
\parbox{0.99\linewidth}{
\begin{equation}
\label{eq::0340}
\begin{aligned}
\ppp(-z^{-1})
=&
(-1)^{\mllll+1} z^{-2 \left(\mllll-1\right)} \ppp(z)
,
\\
\qqq(-z^{-1})
=&
z^{-2\mllll}
\left((-1)^\mllll \mu\, \ppp(z) +z^2 \rrr(z)\right)
,
\\
\rrr(-z^{-1})
=&
z^{-2 (\mllll-1)} \left(\mu z^2 \ppp(z) + \qqq(z)\right)
,
\\
\sss(-z^{-1})
=&
-z^{-2\mllll} \left(
\mu\left(\mu z^2 \ppp(z) + \qqq(z)\right)
+(-1)^\mllll z^2 \left(\mu z^2 \rrr(z) + \sss(z)\right)
\right);
\end{aligned}
\hspace{0.8em}
\end{equation}
}
\\[0.3em]\phantom{.}

 \noindent
\parbox{0.99\linewidth}{
\begin{equation}
\label{eq::0350}
\begin{aligned}
\ppp(-z)
=&\;
(-1)^{\mllll+1} (\lambda + \mu^2)^{-1} \left(\mu z^2 \rrr(z) + \sss(z)\right)
,
\\
\qqq(-z)
=&\;
\mu z^2 \ppp(z) + \qqq(z)
+(-1)^\mllll (\lambda + \mu^2)^{-1} \mu z^2 \left(\mu z^2 \rrr(z) + \sss(z)\right)
,
\\
\rrr(-z)
=&\;
\rrr(z)
,
\\
\sss(-z)
=&\;
(-1)^{\mllll+1} (\lambda + \mu^2) \ppp(z)
- \mu z^2 \rrr(z);
\end{aligned}
\hspace{2.0em}
\end{equation}
}
\\[0.3em]\phantom{.}

\noindent
\parbox{0.99\linewidth}{
\begin{equation}
\label{eq::0360}
\begin{aligned}
\ppp(z^{-1})
=&\;
(\lambda + \mu^2)^{-1} z^{-2 (\mllll-1)} \left(\mu z^2 \rrr(z) + \sss(z)\right)
,
\\
\qqq(z^{-1})
=&\;
(\lambda + \mu^2)^{-1} z^{-2\mllll} \left(\lambda z^2 \rrr(z) - \mu \,\sss(z)\right)
,
\\
\rrr(z^{-1})
=&\;
z^{-2 (\mllll-1)} \left(\mu z^2 \ppp(z) + \qqq(z)\right)
,
\\
\sss(z^{-1})
=&\;
z^{-2\mllll} \left(\lambda z^2 \ppp(z) - \mu \,\qqq(z)\right),
\end{aligned}
\hspace{9.9em}
\end{equation}
}

\noindent
which manifest the characteristic property of the polynomials $\ppp,\!\qqq,\!\rrr,\!\sss$
 \cite{BT2}.

The above relationships suffice for the proving  validity of the rules
 \eqref{eq::0190}-\eqref{eq::0270}, \eqref{eq::0310}-\eqref{eq::0330}.

{\it At second}, the quantity $\delTa$ defined in \eqref{eq::0280}
as a function of $z$ (a power function  times a polynomial) does not actually depend on this
indeterminate. This is proven by means of computation of its derivative 
which turns out to be the identical zero 
as a consequence of \Eqs\ \eqref{eq::0180}. $\delTa$ is therefore a constant
or, more exactly, a function (polynomial) of the parameters $\lambda$
and $\mu$ (depending also on $\mllll$).
Given $\ppp,\!\qqq,\!\rrr,\!\sss$, it can be computed substituting into
the right-hand side of \eqref{eq::0280}
any numerical value of $z$. The most simple case corresponds to $z=1$, obviously.
In this way, we note that the same substitution to the equalities \eqref{eq::0360}
yields the following reductions
\begin{equation}  \label{eq::0370}
\qqq(1)=\rrr(1)-\mu\,\ppp(1),\; \sss(1)=(\lambda+\mu^2)\ppp(1)-\mu\,\rrr(1).
\end{equation}
Applying them, we obtain
\begin{equation}\label{eq::0380}
\delTa=
(\lambda+\mu^2)\ppp(1)^2-\rrr(1)^2
.
\end{equation}

Obviously, the properties of the composition rules \eqref{eq::0190}-\eqref{eq::0270}
in the cases $ \delTa\not=0$ and $ \delTa=0$ drastically differ.
We may consider the latter
of them degenerate as compared to the former and,
limiting our consideration
to a generic situation,
assume to be imposed below
 the restriction
\begin{equation}\label{eq::0390}
\delTa\not=0
.
\end{equation}
treating it as a genericity condition.

Similar remark concerns 
the vanishing of the sum $\lambda+\mu^2$.
Apart of 
affecting the properties
of the rules  \eqref{eq::0190}-\eqref{eq::0270},
it would considerably alter
the meaning of the equalities \eqref{eq::0350},\eqref{eq::0360}. 
The condition 
\begin{equation}
\label{eq::0400}
\lambda+\mu^2\not=0
\end{equation}
is the third  restriction, after inequalities $\mu\not=0$ and \eqref{eq::0390}, which will be assumed
to be fulfilled throughout.

It seems worthwhile to mention here yet another application of the
transformations \eqref{eq::0340},\eqref{eq::0350}.
They can be utilized 
for modification of the formulas
\eqref{eq::0070},\eqref{eq::0080}
by means of the explicit manifestation of the result of
replacements of the variable $z$ which are carried out over their
right-hand expressions. In particular,  the
functions
$\ppp,\!\qqq,\!\rrr,\!\sss$ with modified argument
can be 
expressed
in terms of themselves with the ``standard'' argument $z$.
In this way, the following equivalent
alternative definitions
of the operators $\opmA $ and $\opmB$
can be obtained:
 \begin{eqnarray}
 \label{eq::0410}
&&
 \begin{aligned}
\opmA[{\somE}](\coveredz)
=&\;
(-1)^\mllll
e^{\mu(z+1/z)}
\Bigl[
\ppp(z)
\big({\somE}'(\coveredco{\,-$1/$}\,\coveredz)
-\mu {\somE}(\coveredco{\,-$1/$}\,\coveredz)\big)+
\\&
\hphantom{ (-1)^\mllll e^{-\mu(z+1/z)} \big( } \hspace{-0.9em}
+\big(\qqq(z)+\mu z^2\ppp(z)\big) {\somE}(\coveredco{\,-$1/$}\,\coveredz)
\Bigr],
\end{aligned}
\\
 \label{eq::0420}
 &&
 \begin{aligned}
\opmB[{\somE}](\coveredz)
=&\;
 2\omega\, z^{1-{\mllll}}e^{\mu(z+1/z)}
\Bigl[
\big(\qqq(z)+\mu z^2\ppp(z)\big)
\bigl({\somE}'(\coveredcominus\!\coveredz)-\mu {\somE}(\coveredcominus\!\coveredz)\bigr) +
\\ & \hphantom{\;  2\omega\, z^{1-{\mllll}}e^{\mu(z+1/z)} }
+(\lambda+\mu^2)\ppp(z){\somE}(\coveredcominus\!\coveredz)
\Bigr].
\end{aligned}
\end{eqnarray}

Among $\op$-operators, $\opmC$ is of particular
interest due its inherent
relations to Eq.~\eqref{eq::0010}.
Inter alia,
the following statement holds true.
\begin{proposition}
Let the function $\somE_\zeta$ holomorphic in $\covered{C}^*$
be an eigenfunction of the operator $\opmC$ with
eigenvalue $\zeta\not=0$. Then $\somE_\zeta$ verifies Eq.~\eqref{eq::0010}
with some $\lambda\not=-\mu^2$.
\end{proposition}
\begin{proof}
The {\em identity\/}
\begin{equation}
\label{eq::0430}
  \opmC\circ\opmC[\somE](\coveredz)
  \equiv
  (2\omega)^2(\lambda+\mu^2)\somE(\coveredz)
  -
  (2\omega)^2\mathrm{lhs}\mbox{\eqref{eq::0010}},
\end{equation}
where $\mathrm{lhs}\mbox{\eqref{eq::0010}}$ stands for the
left-hand side expression of Eq.~\eqref{eq::0010},
is established by means of straightforward computation.
If $ \opmC[\somE_\zeta]=\zeta\somE_\zeta$ then
it 
converts to the equality
$(2\omega)^2\mathrm{lhs}\mbox{\eqref{eq::0010}}=
((2\omega)^2(\lambda+\mu^2)-\zeta^2)\somE$. Obviously, there exists $\lambda$
such that it implies
$ \mathrm{lhs}\mbox{\eqref{eq::0010}}=0 $,
i.e.\
Eq.~\eqref{eq::0010} with such $\lambda$ is fulfilled.
\end{proof}
It is worth noting that
the only smooth eigenfunction of $ \opmC $
with {\em zero\/} eigenvalue
 is obviously 
the exponent
$\somE_1(z)=e^{\mu{}z}$
which 
verifies Eq.~\eqref{eq::0010} if and only if $\lambda=-\mu^2$.
On the other hand,
for such 
 $ \lambda$, the function
\begin{equation}\label{eq::0440}
\begin{aligned}
&\somE_2(z) = \exp(\mu z+ \varepsilon(z)),
\\
&\mbox{where }\varepsilon'(z) =z^{\mllll-1} e^{-\mu(z-1/z)}/\tilde\varepsilon(z),
\mbox{ where }\tilde\varepsilon'(z) = z^{\mllll-1} e^{-\mu(z-1/z)},
\end{aligned}
\end{equation}
determined 
by means of two subsequent quadratures of known functions,
also verifies Eq.~\eqref{eq::0010}%
       \footnote{Here we consider the solutions $ \somE_{1},\somE_{2}$ 
       in some vicinity of the unit
       alone and hence may identify therein their `genuine argument'
$\coveredz$ with its projection $z=\iota\,\coveredz$.}.
The two solutions $ \somE_1,\somE_2 $
constitute the basis of the linear space $\mathbf\Omega$ 
and any solution to
Eq.~\eqref{eq::0010} coincide with some their linear combination.
Since under generic conditions  representation of solutions to Eq.~\eqref{eq::0010}
in terms of
explicit quadratures is not possible\footnote{
Another
 known exception is the case of existence of a polynomial solution
which may arise just if the order is a negative integer, see \cite{T2,BT3}.
}, we may treat the case $\lambda=-\mu^2$
as degenerate.
This is another argument in favor of acceptance 
of the
genericity condition \eqref{eq::0400}.

Let now $ \somE $ be any non-trivial solution to Eq.~\eqref{eq::0010}. 
Let us define
the pair of non-coinciding functions
\begin{equation}
\label{eq::0450}
\EEpm{\pm}=\somE \pm (2\omega)^{-1}(\lambda+\mu^2)^{-1/2}\opmC[\somE].
\end{equation}
They both 
satisfy Eq.~\eqref{eq::0010} as well.
Using 
the equation arising from the identity \eqref{eq::0430} 
in the case $ \mathrm{lhs}\mbox{\eqref{eq::0010}}=0 $,
we obtain
\begin{equation}
\label{eq::0460}
\opmC[\EEpm{\pm}]=\opmC[\somE] \pm (2\omega)(\lambda+\mu^2)^{1/2}\somE=
\pm (2\omega)(\lambda+\mu^2)^{1/2}\EEpm{\pm}.
\end{equation}
Thus we may state the following.
           \begin{lemma}
If the condition \eqref{eq::0400}   
is met and
$\somE $
is any non-trivial solution to Eq.~\eqref{eq::0010}
 then the two functions \eqref{eq::0450},
also verifying Eq.~\eqref{eq::0010},
either are the eigenfunctions of the operator $ \opmC $
with eigenvalues $ \pm (2\omega)(\lambda+\mu^2)^{1/2} $, respectively,
or one of them is  the identically zero function.
The latter case takes place if and only if $ \somE $ is itself
an eigenfunction of $ \opmC $; then the corresponding
eigenvalue coincides with
one of the two ones pointed out above.
             \end{lemma}
The next almost obvious 
statement is  
germane 
to the above 
assertion: 
               \begin{lemma}
There exists a solution to Eq.~\eqref{eq::0010} which is not an eigenfunction
of the given operator $ \opmC $.
               \end{lemma}
\begin{proof}
Let us suppose the opposite, i.e.\ let all the non-trivial solutions to the
given Eq.~\eqref{eq::0010} be at the same time eigenfunctions of
$ \opmC $. 
Let, in particular, $\somE_1$ and $\somE_2$ be such solutions which are
linear independent. Due to \eqref{eq::0210} and \eqref{eq::0400}, 
their eigenvalues $\zeta_1$ and $\zeta_2$
are both non-zero. 
We have $ \opmC(\somE_1+\somE_2)=\half(\zeta_1+\zeta_2)(\somE_1+\somE_2 )
+\half(\zeta_1-\zeta_2)(\somE_1-\somE_2)$. 
The sum $ \somE_1+\somE_2 $ verifies Eq.~\eqref{eq::0010}
and then, in accordance with supposition,
$ \opmC(\somE_1+\somE_2)= \zeta (\somE_1+\somE_2)$ for some non-zero constant $\zeta$.
Thus we obtain
$(\zeta_1+\zeta_2-2\zeta) (\somE_1+\somE_2) + (\zeta_1-\zeta_2)  (\somE_1-\somE_2)=0 $.
Since $\somE_1$ and $\somE_2$ are linear independent, both of these coefficients vanish.
Accordingly, it holds
$\zeta_1=\zeta_2=\zeta\not=0$. 
Now, let us mention that the value of the expression
\begin{equation}
       \label{eq::0470}
w=z^{1-\mllll}e^{-\mu(z+1/z)}
(\somE_1'(\coveredz)\somE_2(\coveredz)-\somE_2'(\coveredz)\somE_1(\coveredz)),
\end{equation}
where, as usual, $z=\iota\,\coveredz$,
does not depend on $ \coveredz $.
Indeed, straightforward computation using
\eqref{eq::0010} 
for eliminating of second derivatives
shows that
its derivative vanishes identically. 
The constant $w$ (the appropriately rescaled wronskian 
of Eq.~\eqref{eq::0010}, in fact)
is zero if and only if the solutions $ \somE_1 $ and $ \somE_2 $
are linear dependent.
Under conditions now assumed, 
it is therefore
nonzero.
But the solutions
$ \somE_1 $ and $ \somE_2 $
are
at the same time
the eigenfunctions of $\opmC$, i.e.\ they obey the equation
\begin{equation} \label{eq::0480} %
\somE'(\coveredz)
=\mu\somE(\coveredz)
+(2\omega)^{-1}\zeta\,z^{\mllll-1}\somE(1/\coveredz).
\end{equation}
Using it,
the derivatives $ \somE'_1, \somE'_2 $
can be eliminated 
and definition \eqref{eq::0470}
takes the form
\begin{equation*}
w= -(2\omega)^{-1}\zeta
e^{-\mu(1+1/z)}
( \somE_1(\coveredz)\somE_2(1/\coveredz)
 -\somE_2(\coveredz)\somE_1(1/\coveredz)).
\end{equation*}
But in accordance with \eqref{eq::0160}
$1/\coveredone=\coveredone$,
and the last factor in parentheses
still
vanishes at the point $\coveredz=\coveredone$, at least.
Thus, ultimately, we have come to
a contradiction. The assumption incompatible with the lemma assertion
is therefore
false and the latter
 is proven.
\end{proof}

It is also worth noting
that the fulfillment of Eq.\ \eqref{eq::0480} 
by eigenfunctions of $\opmC$
implies the following statement.
\begin{lemma}
If $\somE$ is an eigenfunction 
the operator
of $ \opmC $ with eigenvalue $\zeta$ then
\begin{equation} \label{eq::0500}
\somE'(\coveredone)
=\big(\mu
+\frac{\zeta}{2\omega}\big)\somE(\coveredone).
\end{equation}
\end{lemma}

Summarizing the above relationships, we obtain the following statement.
                             \begin{theorem}
Let the condition \eqref{eq::0400}
be fulfilled and Eq.~\eqref{eq::0010} be given.
Then the solutions $\EEpm{+},\EEpm{-}$
of the Cauchy problems posed at the point $\coveredz=\coveredone$
with the initial data
obeying the restrictions
              \begin{equation}\label{eq::0510}
\EEpm{\pm}'(\coveredone)
=\big(\mu
\pm(\lambda+\mu^2)^{1/2}\big)\EEpm{\pm}(\coveredone).
\end{equation}
are at the same time the
eigenfunctions of the operator $ \opmC $
with eigenvalues 
$ \pm (2\omega)(\lambda+\mu^2)^{1/2} $, respectively.
They constitute the basis 
of the linear space $\mathbf\Omega$ of 
 solutions to Eq.~\eqref{eq::0010}
 and are unique 
 up to multiplications by constant factors.
 There are no other
 (i.e.\ linear independent with both 
 $\EEpm{+}$ and 
  $\EEpm{-}$, separately)
 eigenfunctions of $ \opmC $ obeying Eq.~\eqref{eq::0010}.
                                \end{theorem}
\begin{corollary}
\begin{equation}\label{eq::0520}
\EEpm{\pm}(\coveredone)\not=0.
\end{equation}
\end{corollary}

Up to now, we considered the parameter $\omega$ scaling
two operators described by \Eqs{~}\eqref{eq::0080} and \eqref{eq::0090}
as an
arbitrary non-zero constant.
Assuming here and below the fulfillment of the condition \eqref{eq::0400}, 
 we may fix it 
 by the following constraint
\begin{equation}
         \label{eq::0530}
(2\omega)^2(\lambda+\mu^2)=1.
\end{equation}
Its advantage is the most simple form of the eigenvalues 
which
the operator $ \opmC $ 
acquires
on the space of solutions to Eq.~\eqref{eq::0010}.
Indeed, in case of fulfillment of \eqref{eq::0530}  
\Eqs{~}\eqref{eq::0460} take the form
\begin{equation}
\label{eq::0540}
\opmC\EEpm{+}=\EEpm{+},\;\opmC\EEpm{-}=-\EEpm{-}.
\end{equation}

We shall derive here one more property of 
eigenfunctions of the operator $ \opmC $
which follows from the above relationships. 
Namely, it follows from
the constancy of the expression \eqref{eq::0470}  
that
\begin{equation*}
z^{1-\mllll}e^{-\mu(z+1/z)}
(\EEpm{+}'(\coveredz)\EEpm{-}(\coveredz)-\EEpm{-}(\coveredz)\EEpm{+}(\coveredz))
=
e^{-2\mu}
(\EEpm{+}'(\coveredone)\EEpm{-}(\coveredone)-\EEpm{-}'(\coveredone)\EEpm{+}(\coveredone)).
\end{equation*}
Besides, in view of 
Eq.\ \eqref{eq::0480}  
in which 
the substitutions
$\somE\leftleftharpoons\EEpm{\pm} $
and 
$\zeta\leftleftharpoons\pm1$
were 
carried out,
one has 
$
\EEpm{+}'(\coveredz)\EEpm{-}(\coveredz)-\EEpm{-}'(\coveredz)\EEpm{+}(\coveredz))
=
(2\omega)^{-1}z^{\mllll-1}
(\EEpm{+}(\coveredz)\EEpm{-}(1/\coveredz)
+\EEpm{-}(\coveredz)\EEpm{+}(1/\coveredz)
).
$ 
 As a consequence,  we obtain
\begin{equation}
\label{eq::0560}
 \EEpm{+}(\coveredz)\EEpm{-}(1/\coveredz)
+\EEpm{-}(\coveredz)\EEpm{+}(1/\coveredz)=
2e^{\mu(z+1/z-2)}\EEpm{+}(\coveredone)\EEpm{-}(\coveredone)
\not=0.
\end{equation}
Let us also specialize Eq.\ \eqref{eq::0560} %
to the cases
$\coveredz= \coveredmcone $
and $\coveredz= \coveredi$.
One has $1/ \coveredmcone =\coveredmpone  $
and $1/\coveredi=\coveredmi  $; therefore, the following 3-point relations
\begin{eqnarray}
\label{eq::0570}
 \EEpm{+}(\coveredmcone)\EEpm{-}(\coveredmpone)
+\EEpm{-}(\coveredmcone)\EEpm{+}(\coveredmpone)&=&
2e^{-4\mu}\EEpm{+}(\coveredone)\EEpm{-}(\coveredone),
\\
\label{eq::0600}
 \EEpm{+}(\coveredi)\EEpm{-}(\coveredmi)
+\EEpm{-}(\coveredi)\EEpm{+}(\coveredmi)&=&
2e^{-2\mu}\EEpm{+}(\coveredone)\EEpm{-}(\coveredone)
\end{eqnarray}
hold true.

%

In accordance with \eqref{eq::0540},
with respect to the basis of the space of solutions
to Eq.~\eqref{eq::0010} 
constituted by the functions $\EEpm{+}, \EEpm{-}$,
the operator $\opmC$ is represented by the 
diagonal matrix $\diag(+1,-1)$. 
It would be useful 
to possess the corresponding matrix representations of the operators  
$\opmA, \opmB$ as well. Their derivation is equivalent to computation
the expansions of the four solutions
$\{\opmA,\opmB\}\times\{\EEpm{+},\EEpm{-}\}$ with respect to the
above basis. 
To that end, beginning with Eqs.\ \eqref{eq::0410}, \eqref{eq::0420},
we eliminate the derivatives $\somE'$ with the help of Eq.\ \eqref{eq::0480}
which in case of fulfillment of Eq.\ \eqref{eq::0530} reads
\begin{equation} \label{eq::0610}  
\EEpm{\pm}'(\coveredz)
=\mu\EEpm{\pm}(\coveredz)
\pm(2\omega)^{-1}z^{\mllll-1}\EEpm{\pm}(1/\coveredz).
\end{equation}
In this way we obtain
\begin{eqnarray}
                  \label{eq::0620}
\hspace{-1.4em}
\opmA[\EEpm{\pm}](\coveredz)
&\!\!\!=\!\!\!&
e^{\mu(z+1/z)}
\big(
\mp
(2\omega)^{-1}
z^{1-\mllll}
\ppp(z)\cdot
 \Mono^{-1} \EEpm{\pm}(\coveredco{-}\!\coveredz 
 )
\\&&\hspace{4.7em}
+(-1)^\mllll
(\qqq(z)+\mu z^2\ppp(z))\cdot
\EEpm{\pm}(\coveredco{\,-$1/$}\,\coveredz 
)
\big),
\nonumber
\\ 
\label{eq::0630}
\hspace{-1.4em}
\opmB[\EEpm{\pm}](\coveredz)
&\!\!\!=\!\!\!&
e^{\mu(z+1/z)}
\big(
(2\omega)^{-1}
z^{1-\mllll}
\ppp(z)\cdot
\EEpm{\pm}(\coveredco{-}\!\coveredz 
)
\\&&\hspace{4.7em}
\mp
(-1)^\mllll
(\qqq(z)+\mu z^2\ppp(z))\cdot
 \Mono^{-1} \EEpm{\pm}(\coveredco{\,-$1/$}\,\coveredz 
 )
\big),
\nonumber
\end{eqnarray}
 We further utilize the following algebraic
{\it identity}
\begin{equation}\label{eq::0640}
\begin{aligned}
e^{\mu(z+1/z)}
\big[
         (2\omega)^{-1}\,z^{1-\mllll}\ppp(z)\, \mathcal{A}
         \mp{}(-1)^\mllll \,(\qqq(z)+\mu z^2\ppp(z))\,\mathcal{B}
\big]\equiv
\\
&\hspace{-24.8em}
{e^{2\mu}}
{(4\omega \somE_1(\coveredone)\somE_2(\coveredone))^{-1}}  \times
\\
&\hspace{-23.8em}
\big[
(
-z^{1 - \mllll}\ppp(z)\, \mathcal{A}
\pm 2\omega (-1)^\mllll  (\qqq(z)+\mu z^2\ppp(z)) \mathcal{B})
\cdot \big\{\!\big\{\somE_1,\somE_2 \big\}\!\big\}\!(\coveredz,1/\coveredz)
\\&  \hspace{-23.8em}             
+\mathbf{W}_{\!\{+\}}[ \mathcal{A}, \pm\mathcal{B};\somE_2](\coveredz,1/\coveredz)
      \somE_1(\coveredz)
+\mathbf{W}_{\!\{-\}}[ \mathcal{A},  \pm\mathcal{B};\somE_1](\coveredz,1/\coveredz)
      \somE_2(\coveredz)
\big]
\end{aligned}
\end{equation}
which is verified by straightforward computation.
Here
$z=\iota\,\coveredz$, $\somE_1 $ and $\somE_2$ stand for some functions holomorphic in $\covered{C}^*$,
$ \mathcal{A} $ and $ \mathcal{B} $ are arbitrary expressions (they cancel out), and
the following abbreviated notations 
are  utilized
\begin{equation} 
\label{eq::0650}
\big\{\!\big\{\somE_1,\somE_2 \big\}\!\big\}(\coveredz,\tilde\coveredz)
= \somE_1(\coveredz) \somE_2(\tilde\coveredz) + \somE_2(\coveredz) \somE_1(\tilde\coveredz)
 - 2 e^{\mu(z+1/z-2)}\somE_1(\coveredone)\somE_2(\coveredone),
\end{equation}
\begin{equation}
\label{eq::0660}
\hspace{-0.06em}
\begin{aligned}
\mathbf{W}_{\{\pm\}}[ \mathcal{A},  \mathcal{B};\somE](\coveredz,\tilde\coveredz)
=&
\pm
2\omega\,
(-1)^\mllll
\big(q(1/z)+\mu/z^{2}\ppp(1/z)\big)\,
\mathcal{A}\cdot \somE (\coveredz)
\\&
-2\omega
(-1)^\mllll
\big(q(z)+\mu\,z^2\ppp(z)\big)\,
\mathcal{B}\cdot \somE (\tilde\coveredz)
\\&
\mp
z^{\mllll-1}\ppp(1/z)
\mathcal{B}\cdot \somE (\coveredz)
+z^{1 - \mllll}\ppp(z)
\mathcal{A}\cdot \somE (\tilde\coveredz)
\\[0.5ex]
=&
\pm
2\omega\,
(-1)^\mllll
z^{2(1-\mllll)}\rrr(z)\, 
\mathcal{A}\cdot \somE (\coveredz)
\\&
-2\omega
(-1)^\mllll
z^{2(\mllll-1)}\rrr(1/z)\, 
\mathcal{B}\cdot \somE (\tilde\coveredz)
\\&
\mp
z^{\mllll-1}\ppp(1/z)
\mathcal{B}\cdot \somE (\coveredz)
+z^{1 - \mllll}\ppp(z)
\mathcal{A}\cdot \somE (\tilde\coveredz),
%
\end{aligned}
\end{equation}
the last equality taking place in view of  \Eqs{~}\eqref{eq::0360}.
If one substitutes into \eqref{eq::0640}
 $\mathcal{A}\leftleftharpoons
\mp\Mono^{-1} \EEpm{\pm}(\coveredrevco\!\coveredz)$,
$\mathcal{B} \leftleftharpoons \mp\EEpm{\pm}(\coveredmiv\,\coveredz )$
then its left-hand side
coincided with the right-hand side of \eqref{eq::0620}.
Similarly, after substitutions
 $\mathcal{A}\leftleftharpoons
 \EEpm{\pm}(\coveredrevco\!\coveredz)$ and
$\mathcal{B} \leftleftharpoons
\Mono^{-1}    \EEpm{\pm}(\coveredmiv\,\coveredz )$
 the left-hand side of \eqref{eq::0640}
coincides with the right-hand side of \eqref{eq::0630}.
Thus we may replace the right-hand sides of \eqref{eq::0620}
and \eqref{eq::0630} with the corresponding instances of \eqref{eq::0640}
in which we have also carried out the substitutions $\somE_1\leftleftharpoons\EEpm{+},
\somE_2\leftleftharpoons\EEpm{-},
   $
Before the writing down the result, it is important to note that in view of the above
selection of the functions $\somE_1, \somE_2$, the factor
$\{\{\somE_1, \somE_2\}\} $
in \eqref{eq::0640}
(and hence the whole first summand therein on the right)
 vanishes due to Eq.~\eqref{eq::0560}.
Accordingly, only $\mathbf W$-involving contributions remain and the results read
\begin{equation}\label{eq::0670}
\begin{aligned}
\opmA[\EEpm{\pm}](\coveredz)
=&
\mp{e^{2\mu}}
{(4\omega \EEpm{+}(\coveredone)\EEpm{-}(\coveredone))^{-1}}  \times
\\
&\hspace{0.7em}
\left(\hspace{0.3em}
\mathbf{W}_{\!\{+\}}
[ \mathcal{M}^{-1}\EEpm{\pm}(\coveredrevco\!\coveredz),\pm \EEpm{\pm}(\coveredmiv\,\coveredz)
;\EEpm{-}](\coveredz,1/\coveredz)
      \cdot\EEpm{+}(\coveredz)+
\right.
\\[-0.5em]&\hspace{0.7em}
\left.
+\mathbf{W}_{\!\{-\}}
[ \mathcal{M}^{-1}\EEpm{\pm}(\coveredrevco\!\coveredz),\pm \EEpm{\pm}(\coveredmiv\,\coveredz)
;\EEpm{+}](\coveredz,1/\coveredz)
     \cdot \EEpm{-}(\coveredz)
\right),
\end{aligned}
\end{equation}
\begin{equation}\label{eq::0680}
\begin{aligned}
\opmB[\EEpm{\pm}](\coveredz)
=&
{e^{2\mu}}
{(4\omega \EEpm{+}(\coveredone)\EEpm{-}(\coveredone))^{-1}}  \times
\\
&\hspace{0.7em}
\left(\hspace{0.3em}
\mathbf{W}_{\!\{+\}}
[ \EEpm{\pm}(\coveredrevco\!\coveredz),\pm \mathcal{M}^{-1}\EEpm{\pm}(\coveredmiv\,\coveredz)
;\EEpm{-}](\coveredz,1/\coveredz)
    \cdot  \EEpm{+}(\coveredz)+
\right.
\\[-0.3em]&\hspace{0.7em}
\left.
+\mathbf{W}_{\!\{-\}}
[ \EEpm{\pm}(\coveredrevco\!\coveredz),\pm \mathcal{M}^{-1}\EEpm{\pm}(\coveredmiv\,\coveredz)
;\EEpm{+}](\coveredz,1/\coveredz)
    \cdot  \EEpm{-}(\coveredz)
\right).
\end{aligned}
\end{equation}
Now, we have to note the following important property which the 
$\mathbf{W}$-functional possess. 
      \begin{theorem}
The values of the functionals
$\mathbf{W}_{\!\{+\}}$ and $\mathbf{W}_{\!\{-\}}$ with
arguments specified in \Eqs{}
\eqref{eq::0670} and \eqref{eq::0680}
do not depend on the value of 
$\coveredz$.
      \end{theorem}
                \begin{proof}
It reduces to computation of the derivative 
of the corresponding expressions.
In them, the derivatives $\EEpm{\pm}'$ are eliminated with the help of 
Eq.~\eqref{eq::0610},
the derivatives 
$\ppp',\!\rrr'$ are eliminated with the help of \Eqs~\eqref{eq::0180},
and the functions $\ppp,\!\qqq,\!\rrr,\!\sss$
with the argument $1/\coveredz$ are expressed in terms of the
same functions with the argument $\coveredz$
with the help of \Eqs~\eqref{eq::0360}.
In all the four cases, 
the results of these transformations
prove to be reducible to identical zero.
                    \end{proof}
Thus \Eqs{~}\eqref{eq::0670}, \eqref{eq::0680} may be interpreted as
the decompositions with constant coefficients of 
their left-hand sides with respect to the basis $\{\EEpm{+},\EEpm{-}\}$.
We may bring them to ``more explicit'' forms by means of fixation of $\coveredz$
in expressions of the coefficients, i.e.\ in the factors 
$\mathbf{W}{\!\{\pm\}}[\cdots](\coveredz,1/\coveredz)$. 
The most simple results arise if $\coveredz\leftleftharpoons\coveredone$. They  read
\begin{equation}\label{eq::0690}
\begin{aligned}
\opmA[\EEpm{ \pm }](\covered z)
&=
\frac{e^{2\mu}}{4\omega}
\Big(
\hspace{1.3em}
\frac{\delTa_{+}}{\EEpm{+}(\coveredone)  }
(\EEpm{\pm}(\coveredrevco\!\coveredone 
) \mp \EEpm{\pm}(\coveredrevpro\!\coveredone 
))
\cdot
\EEpm{+}(\covered{z})
\\&\hphantom{  = \frac{e^{2\mu}}{4\omega}   \big( \,\, }
-\frac{\delTa_{-}}{\EEpm{-}(\coveredone) }
(\EEpm{\pm}(\coveredrevco\!\coveredone 
)  \pm  \EEpm{\pm}(\coveredrevpro\!\coveredone 
))
\cdot
\EEpm{-}(\covered{z})
\Big),
\end{aligned}
\end{equation}
\begin{equation}\label{eq::0700}
\begin{aligned}
\opmB[\EEpm{ \pm }](\covered z)
&=
\frac{e^{2\mu}}{4\omega}
\Bigl(
\hspace{1.3em}
\frac{\delTa_{+}}{\EEpm{+}(\coveredone)  }
(\EEpm{\pm}(\coveredrevco\!\coveredone 
) \mp \EEpm{\pm}(\coveredrevpro\!\coveredone 
))
\cdot
\EEpm{+}(\covered{z})
\\&\hphantom{  = \frac{e^{2\mu}}{4\omega}   \big( \,\, }
+\frac{\delTa_{-}}{\EEpm{-}(\coveredone) }
(\EEpm{\pm}(\coveredrevco\!\coveredone 
)  \pm  \EEpm{\pm}(\coveredrevpro\!\coveredone 
))
\cdot
\EEpm{-}(\covered{z})
\Bigr),
\end{aligned}
\end{equation}
\begin{equation}
 \label{eq::0710}
\mbox{where }
\delTa_{\pm}=
\ppp(1) \pm (-1)^\mllll2\omega\,\rrr(1).
\end{equation}
Let us note that in view of Eq.\ \eqref{eq::0530} and in accordance with 
expanded definition  
\eqref{eq::0380} of $\delTa  $ it holds
\begin{equation}
 \label{eq::0720}
\delTa_{+}\delTa_{-}=(2\omega)^{-2}\delTa.
\end{equation}
Hence the condition \eqref{eq::0390} can also be represented as 
the union of the two inequalities
\begin{equation}
 \label{eq::0730}
\delTa_{+}\not=0\not=\delTa_{-}
\end{equation}
which remove 
two algebraic curves from the set of values of the parameters $\lambda,\mu$.

A direct consequence of \Eqs{~}\eqref{eq::0690},\eqref{eq::0700} are the following 
explicit
{\it
matrix representations\/} $\matrA, \matrA $
of the operators $\opmA, \opmB $ with respect to the basis
$\{\EEpm{+},\EEpm{-}\}$
\begin{eqnarray}
                 \label{eq::0740}
                 &&
\begin{aligned}
\matrA  
&=
\frac{e^{2\mu}}{4\omega}
\begin{pmatrix}
\EEpm{+}(\coveredrevco\!\coveredone)  -  \EEpm{+}(\coveredrevpro\!\coveredone)
&\hspace{-0.5em}
-(\EEpm{+}(\coveredrevco\!\coveredone)    +   \EEpm{+}(\coveredrevpro\!\coveredone))
\\
\EEpm{-}(\coveredrevco\!\coveredone)   +   \EEpm{-}(\coveredrevpro\!\coveredone)
&\hspace{-0.5em}
-(\EEpm{-}(\coveredrevco\!\coveredone)  -   \EEpm{-}(\coveredrevpro\!\coveredone))
\end{pmatrix}
\mathrm{D}^+_-,
\end{aligned}
%
\\ &&
         \label{eq::0750}
\begin{aligned}
\matrB  
&=
\frac{e^{2\mu}}{4\omega}
\begin{pmatrix}
\EEpm{+}(\coveredrevco\!\coveredone)  -  \EEpm{+}(\coveredrevpro\!\coveredone)
&\hspace{-0.5em}
\EEpm{+}(\coveredrevco\!\coveredone)    +   \EEpm{+}(\coveredrevpro\!\coveredone)
\\
\EEpm{-}(\coveredrevco\!\coveredone)   +   \EEpm{-}(\coveredrevpro\!\coveredone)
&\hspace{-0.5em}
\EEpm{-}(\coveredrevco\!\coveredone)  -   \EEpm{-}(\coveredrevpro\!\coveredone)
\end{pmatrix}
\mathrm{D}^+_-
\end{aligned}
%
\\
&&
         \label{eq::0760}
\mbox{ where }
\mathrm{D}^+_-=\diag
\left(
\frac{\delTa_{+}}{\EEpm{+}(\coveredone)}
,
\frac{\delTa_{-}}{\EEpm{-}(\coveredone)}
\right).
\end{eqnarray}

Now let  us mention that in accordance with the composition rule 
\eqref{eq::0190} the matrix representation of the operator $\opmA$
has to fulfill the equation $\matrA^2=-\delTa\diag({1,1})$. 
A straightforward computation
with the matrix \eqref{eq::0740} shows that
this is the case if and only if
\begin{equation}
                 \label{eq::0770}
\frac{\delTa_{+}}{\EEpm{+}({\coveredone}) }
\big(\EEpm{+}({\coveredrevco\!\coveredone})-\EEpm{+}({\coveredrevpro\!\coveredone})\big)
=
\frac{\delTa_{-}}{\EEpm{-}({\coveredone})}
\big(\EEpm{-}({\coveredrevco\!\coveredone})-\EEpm{-}({\coveredrevpro\!\coveredone})\big).
\end{equation}%
This is yet another relationship constraining 
the values of the
eigenfunctions of the operator $\opmC$,
this time
at the points projected to $-1$ and $1$.
Applying it, the matrices $\matrA, \matrB$ can be transformed to
the following representations
\begin{eqnarray}
                 \label{eq::0780}
&&
\begin{aligned}
\matrA   
&=
\frac{e^{2\mu}}{4\omega}
\mathrm{D}^-_+ 
\begin{pmatrix}
\EEpm{-}(\coveredrevco\!\coveredone)  -  \EEpm{-}(\coveredrevpro\!\coveredone)
&
-(\EEpm{+}(\coveredrevco\!\coveredone)    +   \EEpm{+}(\coveredrevpro\!\coveredone))
\\
\EEpm{-}(\coveredrevco\!\coveredone)   +   \EEpm{-}(\coveredrevpro\!\coveredone)
&
-(\EEpm{+}(\coveredrevco\!\coveredone)  -   \EEpm{+}(\coveredrevpro\!\coveredone))
\end{pmatrix},
\end{aligned}
%
\\
                \label{eq::0790}
&&\begin{aligned}
\matrB  
&=
\frac{e^{2\mu}}{4\omega}
\mathrm{D}^-_+ 
\begin{pmatrix}
\EEpm{-}(\coveredrevco\!\coveredone)  -  \EEpm{-}(\coveredrevpro\!\coveredone)
&
\EEpm{+}(\coveredrevco\!\coveredone)    +   \EEpm{+}(\coveredrevpro\!\coveredone)
\\
\EEpm{-}(\coveredrevco\!\coveredone)   +   \EEpm{-}(\coveredrevpro\!\coveredone)
&
\EEpm{+}(\coveredrevco\!\coveredone)  -   \EEpm{+}(\coveredrevpro\!\coveredone)
\end{pmatrix}=\matrA\diag(1,-1),
\end{aligned}
%
\\
                \label{eq::0800}
                &&
\mbox{ where }
\mathrm{D}^-_+=\diag
\left(
\frac{\delTa_{-}}{\EEpm{-}(\coveredone)}
,
\frac{\delTa_{+}}{\EEpm{+}(\coveredone)}
\right).
\end{eqnarray}
Straightforward computations taking into account Eq.~\eqref{eq::0570}
shows that
\begin{equation}
\det\matrA = \delTa,\;\det\matrB = -\delTa.
\end{equation}
In view of the requirement \eqref{eq::0390}, the matrices $\matrA, \matrB$
are always invertible. The same holds true for the operators $\opmA, \opmB$.
Thus, we have the following statement.
\begin{theorem}
Under restrictions currently assumed,
$\op$-operators
determine linear
automorphisms of the space of solutions to \eqref{eq::0010}. 
\end{theorem}\noindent
(In case of the operator $\opmC$, this property follows from
definitions and is, in fact, evident.)

The composition rule \eqref{eq::0190} of $\op$-operators have gained us 
the constraint \eqref{eq::0770} characterizing certain variations of 
the eigenfunctions of $\opmC$. In a similar way, converting the rule
\eqref{eq::0200}, which now reads 
$\opmB\circ\opmB=\delTa  \Mono$,
to its matrix form, 
the matrix 
 $\mathbf{M}=\delTa^{-1} \matrB^2 $
of the \monodromy{} transformation $\Mono $
can be derived.
Specifically, using the representation \eqref{eq::0790} and Eq. \eqref{eq::0770},
 one obtains
\begin{equation}
             \label{eq::0810}
\hspace{-0.3em}
\begin{aligned}
&\mathbf{M}=
 e^{4\mu}\big(2\EEpm{+}(\coveredone)\EEpm{-}(\coveredone) \big)^{-1}\times
\\
&
\begin{pmatrix}
\EEpm{+}(\coveredrevco\!\coveredone)\EEpm{-}(\coveredrevco\!\coveredone)
+
\EEpm{+}(\coveredrevpro\!\coveredone)\EEpm{-}(\coveredrevpro\!\coveredone)
&
\EEpm{+}(\coveredrevco\!\coveredone)^2-\EEpm{+}(\coveredrevpro\!\coveredone)^2
\\
\EEpm{-}(\coveredrevco\!\coveredone)^2-\EEpm{-}(\coveredrevpro\!\coveredone)^2
&\hspace{-3em}
\EEpm{+}(\coveredrevco\!\coveredone)\EEpm{-}(\coveredrevco\!\coveredone)
+
\EEpm{+}(\coveredrevpro\!\coveredone)\EEpm{-}(\coveredrevpro\!\coveredone)
\end{pmatrix}.
\end{aligned}
\end{equation}
               Eq.~\eqref{eq::0570} then implies that
\begin{equation}
             \label{eq::0820}
\det \mathbf{M}
=1, \mbox{ and, therefore, }
\mathbf{M}^{-1}=\diag(1,-1)\,\mathbf{M}\diag(1,-1)
\end{equation}
is produced
from $ \mathbf{M} $ 
by means of the reversing signs of its antidiagonal elements.

The \monodromy{} transformation enables one to carry out analytic continuation
of solutions to Eq.\ \eqref{eq::0010}
from certain \subdomain{} to the whole domain $ \covered{C}^* $
by means of algebraic operations. 
It suffices to realize such continuation
for the solutions $\EEpm{\pm}$ constituting the basis of $\mathbf\Omega$. 
To that end, let us note
that if we introduce 
the two-element column $\mathbf\somE=(\EEpm{+},\EEpm{-})^\mathrm{T}$,
then the gist  
of the matrix \eqref{eq::0810} 
 is exactly to serve
the operator of the transformation
\begin{equation}\label{eq::0830}
\mathbf{\somE}(\mathrm{M}^k\coveredz)=\mathbf{M}^k\,\mathbf{\somE}( \coveredz), \;
k\in \mathbb{Z}.
\end{equation}
As far as one concerns the left-hand side, the map 
$ \mathbf{\somE}(\coveredz)\to \mathbf{\somE}(\mathrm{M}^k\coveredz) $ 
can be interpreted as certain ``shift'' of a subset of the domain 
where $ \mathbf{\somE} $ is evaluated.
Specifically, let 
the values of 
$ \mathbf{\somE}(\coveredz) $ be known (``have been computed'')
everywhere in the \subdomain{}
represented in frames of semilog model of $ \covered{C}^* $
by the semi-strip $\covered{C}^*_0:=\mathbb{R}_+ \times (-\pi,\pi)$
which projects bijectively to $\mathbb{C}^*\fgebackslash\mathbb{R}_-$.
The map $\coveredz\to \mathrm{M}^k\coveredz$ defined by Eq.\ \eqref{eq::0170}
takes $\covered{C}^*_0$ 
to 
the semi-strip $\covered{C}^*_k:=\mathbb{R}_+ \times (\pi k-\pi,\pi k+\pi)$.
Accordingly,
one can determine the value of $ \mathbf{\somE}$ at any point of 
$ \covered{C}^*_k $ from its value at certain point of $ \covered{C}^*_0 $
by means of Eq.\ \eqref{eq::0830}.
Since the union of all $ \covered{C}^*_k$ covers the whole $ \covered{C}^*$
but the rays $(\forall \rho>0,\pi k) $, 
where  $ \mathbf{\somE}$  still can be determined by continuity,
its value can be 
found 
in this way at any point 
of
the whole domain $ \covered{C}^*$ by means of an algebraic transformation.

The making use of the \monodromy{} transformation
for continuation of solutions from a \subdomain{} with closure of projection
coinciding with the complex plane is not specific for Eq.\ \eqref{eq::0010}.
However, in our case, there exists the opportunity to carry out 
continuation from a \subdomain{}  which is ``twice less''. Namely, we shall show
that the values on the eigenfunctions $\EEpm{\pm}$ can be determined 
by means of algebraic operations
at any point of their domain from their values in a \subdomain{}
projected to the semi-plane $\Re z>0$.
The corresponding transformations are derived from the
composition rules  \eqref{eq::0200}, \eqref{eq::0220} in which
one $\op$-operator is 
represented by the formula  \eqref{eq::0070} or  \eqref{eq::0080} 
while another one is taken in a matrix representation, see
\Eqs{~}\eqref{eq::0640}, \eqref{eq::0750}. They read
\begin{eqnarray}
\label{eq::0840}           
\EEpm{\pm}(\coveredrevco\!\coveredz 
) &=&\;
\half{}z^{\mllll-1}
\,e^{\mu(2-z-1/z)}\times
\\&&\hspace{-6.65em}
\bigg(
(-1)^\mllll2\omega\,z^{\mllll-1}
\rrr(1/z)
  \Big(
-\frac{ \EEpm{\pm}(\coveredrevco\!\coveredone) \mp \EEpm{\pm}(\coveredrevpro\!\coveredone) }%
     { \llap{$\delTa_{-}\,$}  \EEpm{+}\rlap{$(\coveredone) $}
     }
\cdot \EEpm{+}(1/\coveredz)
+ 
\frac{ \EEpm{\pm}(\coveredrevco\!\coveredone) \pm \EEpm{\pm}(\coveredrevpro\!\coveredone) }%
     {  \llap{$\delTa_{+}\,$}  \EEpm{-}\rlap{$(\coveredone) $}
     }
\cdot \EEpm{-}(1/\coveredz)
   \Big)+
\nonumber\\&& \hspace{-4em} 
+
\ppp(1/z)
    \Big(
\frac{ \EEpm{\pm}(\coveredrevco\!\coveredone) \mp \EEpm{\pm}(\coveredrevpro\!\coveredone) }%
     { \llap{$\delTa_{-}\,$} \EEpm{+}\rlap{$(\coveredone) $}  
}
\cdot \EEpm{+}(\coveredz)
+ 
\frac{ \EEpm{\pm}(\coveredrevco\!\coveredone) \pm \EEpm{\pm}(\coveredrevpro\!\coveredone) }%
     {\llap{$\delTa_{+}\,$} \EEpm{-}\rlap{$(\coveredone) $}  
}
\cdot \EEpm{-}(\coveredz)
 \Big)\bigg),
\nonumber
\\
%
%
%
\label{eq::0850}
\EEpm{\pm}(\coveredrevpro\!\coveredz 
) &=&\;
\pm\half{}z^{\mllll-1}
\,e^{\mu(2-z-1/z)}\times
\\&&\hspace{-6.65em}
\bigg(
(-1)^\mllll2\omega\,z^{\mllll-1}
\rrr(1/z)
  \Big(
 \frac{ \EEpm{\pm}(\coveredrevco\!\coveredone) \mp \EEpm{\pm}(\coveredrevpro\!\coveredone) }%
     { \llap{$\delTa_{-}\,$}  \EEpm{+}\rlap{$(\coveredone) $}
     }
\cdot \EEpm{+}(1/\coveredz)
 + 
\frac{ \EEpm{\pm}(\coveredrevco\!\coveredone) \pm \EEpm{\pm}(\coveredrevpro\!\coveredone) }%
     {  \llap{$\delTa_{+}\,$}  \EEpm{-}\rlap{$(\coveredone) $}
     }
\cdot \EEpm{-}(1/\coveredz)
   \Big) +
\nonumber\\&& \hspace{-4em} 
+
\ppp(1/z)
    \Big(
-\frac{ \EEpm{\pm}(\coveredrevco\!\coveredone) \mp \EEpm{\pm}(\coveredrevpro\!\coveredone) }%
     { \llap{$\delTa_{-}\,$} \EEpm{+}\rlap{$(\coveredone) $}  
}
\cdot \EEpm{+}(\coveredz)
+ 
\frac{ \EEpm{\pm}(\coveredrevco\!\coveredone) \pm \EEpm{\pm}(\coveredrevpro\!\coveredone) }%
     {\llap{$\delTa_{+}\,$} \EEpm{-}\rlap{$(\coveredone) $}  
}
\cdot \EEpm{-}(\coveredz)
\Big)\bigg),
\nonumber
\end{eqnarray}

\noindent
Let us remind the meaning of some notations utilized above. Here
$\coveredz$ denotes an arbitrary point of the Riemann surface
$\covered{C}^*$ and 
$z=\iota\,\coveredz$ is its projection to $\mathbb{C}^*=\mathbb{C}\fgebackslash0$
for which  $\covered{C}^*$ serves the universal cover (each point of the former is lifted to 
the two-side sequence of points from the latter).
The points 
$\coveredrevco\!\coveredz$
and 
$\coveredrevpro\!\coveredz$ belonging to $\covered{C}^*$
both project to $-z$. They are the boundary points of the non-closed arc 
projecting bijectively to the ``punctured circle'' 
$\puncturedS_{(-z)}=
\{\tilde z:\tilde z\in\mathbb{C},|\tilde z|= |z|, \tilde z\not=-z\}$, 
in the middle of which $\coveredz$
is situated. The points $\coveredrevco\!\coveredone$ and $\coveredrevpro\!\coveredone$,
projecting to $-1$,
represent 
the particular case of $\coveredrevco\!\coveredz$
and
$\coveredrevpro\!\coveredz$ 
arising if $\coveredz=\coveredone$. The inversion $\coveredz\to1/\coveredz $
is understood in accordance with the rule \eqref{eq::0160}. Correspondingly, the
``right half'' of $\covered{C}^*_{0} $,
the \subdomain{}
$ \half\covered{C}^*_{0} $ modeled 
by the half-strip $\mathbb{R}_+\times (-\pi/2,\pi/2)$
and projecting to the half-plane $\Re \tilde z >0 $, 
is invariant with respect to it.
It means that,
varying $\coveredz$ within $ \half\covered{C}^*_{0} $, 
the arguments of the functions $\EEpm{\pm}$
in the right-hand sides of \Eqs{~}\eqref{eq::0840} and \eqref{eq::0850}
remain well within the same \subdomain{}.
At the same time, the corresponding \subdomain{}s,
to which the arguments of the 
left-hand side functions of \Eqs{~}\eqref{eq::0840} and \eqref{eq::0850} 
belong,
are distinct; 
in particular, they do not intersect with $ \half\covered{C}^*_{0} $.
Indeed, it is easy to see that 
in case of the equation \eqref{eq::0840} 
the argument of the functions $\EEpm{\pm}$ on the left runs through the
\subdomain{} modeled by the half-strip $\mathbb{R}_+\times (\pi/2,3\pi/2)$
while 
for the left-hand side function in Eq.~\eqref{eq::0850} 
the \subdomain{} where it is evaluated
is modeled by
the half-strip $\mathbb{R}_+\times (-3\pi/2,-\pi/2)$, these maps being 1-to-1. 
Thus,
with the help of \Eqs{~}\eqref{eq::0840}, \eqref{eq::0850}, starting from $\half\covered{C}^*_{0} $,
one can continue 
both functions $\EEpm{\pm}$ to the \subdomain{}
modeled by the half-strip $\mathbb{R}_+\times (-3\pi/2,3\pi/2)$ which is ``three times wider''
than the original one.
In particular, it embodies 
$\covered{C}^*_{0} $ and, hence,  applying further 
the \monodromy{} transformation in a way discussed above, the functions
$\EEpm{\pm}$ can be continued to their whole domain $\covered{C}^*$,
all the operations involved in such transformations being algebraic.

It is important to note, however, that 
\Eqs{~}\eqref{eq::0840} and \eqref{eq::0850},  as they stand, 
can not be considered as the ready-to-be-used tools for continuation of the
eigenfunctions $\EEpm{\pm}$ from their \subdomain{} $\half\covered{C}^*_{0} $.
The point is that these contain 
the quantities 
$\EEpm{\pm}(\coveredrevco\!\coveredone)$ and $\EEpm{\pm}(\coveredrevpro
\!\coveredone)$
where the functions  $\EEpm{\pm}$ are evaluated at the points which belong neither to the
\subdomain{} from which the continuation has to be carried out, nor to
its boundary. 
Yet, the difficulty can be settled
as follows.

Let us evaluate Eq.~\eqref{eq::0840} at the point $\coveredz=\coveredmi$
and Eq.~\eqref{eq::0850} at the point $\coveredz=\coveredi$.
An 
inspection shows that in such cases
all the functions $\EEpm{\pm}$ of arguments depending on $\coveredz$
involved in these equations prove to be evaluated 
just at these two points. But the both points $\coveredi$ and $\coveredmi$
belong to the boundary of 
$\half\covered{C}^*_{0} $.
Hence, by continuity, the values of $\EEpm{\pm}( \coveredi)$ and $\EEpm{\pm}( \coveredmi)$
may be considered to be known. Then the two pairs of the above equations 
may be treated as the closed inhomogeneous linear system with four unknowns 
$\EEpm{\pm}(\coveredrevco\!\coveredone)$ and $\EEpm{\pm}(\coveredrevpro\!\coveredone)$.
It proves to be solvable 
leading to the following formulas 
\begin{eqnarray}
\label{eq::0860}
\EEpm{\pm}(\coveredrevco\!\coveredone 
) + \EEpm{\pm}(\coveredrevpro\!\coveredone 
)
&=&
(\delTa_{\mp}\EEpm{\pm}(\coveredone))^{-1}\times
\\&&\hspace{-6em}
\bigg(%
-\Imi^{\mllll+1}\big(\ppp(-\Imi)\EEpm{\pm}^2(\coveredi)
-(-1)^\mllll(\ppp(\Imi)\EEpm{\pm}^2(\coveredmi)
\big)
\nonumber\\[-1.1em]&&\hspace{-5em}
\mp(-1)^{\mllll}2\omega
\big(
\qqq(\Imi)+\qqq(-\Imi)-\mu(\ppp(\Imi)+\ppp(-\Imi))
\big)
\EEpm{\pm}(\coveredi)\EEpm{\pm}(\coveredmi)
\bigg),
\nonumber
\\
\label{eq::0870}
\EEpm{\pm}(\coveredrevco\!\coveredone 
) - \EEpm{\pm}(\coveredrevpro\!\coveredone 
)
&=& 
(\delTa_{\pm}\EEpm{\mp}(\coveredone))^{-1}\times
\\&&\hspace{-10em}
 \bigg(
 -\Imi^{\mllll+1}
\big(\ppp(-\Imi)\EEpm{+}(\coveredi)\EEpm{-}(\coveredi)
 +(-1)^\mllll
      \ppp(\Imi)\EEpm{+}(\coveredmi)\EEpm{-}(\coveredmi)
 \big)
\nonumber\\[-0.9em]&&\hspace{-11em}
+(-1)^{\mllll}2\omega
\big(
(
 \qqq(\mp\Imi)-\mu\,\ppp(\mp\Imi)
 )
 \EEpm{+}(\coveredi)\EEpm{-}(\coveredmi)
-
(
 \qqq(\pm\Imi)-\mu\,\ppp(\pm\Imi)
 )
 \EEpm{+}(\coveredmi)\EEpm{-}(\coveredi))
\big)
 \bigg).
 \nonumber
\end{eqnarray}
These allow one to eliminate
their left-hand side sums and differences from the
right-hand sides of 
\Eqs{~}\eqref{eq::0840}, \eqref{eq::0850},
exempting 
them from presence of quantities unknown in advance.
\Eqs{~}\eqref{eq::0840}, \eqref{eq::0850}
in conjunctions with \Eqs{~}\eqref{eq::0860}, \eqref{eq::0870}
and Eq.\ \eqref{eq::0830}
enable one to carry out analytic continuation of
eigenfunctions of the operator $\opmC$
from the \subdomain{}
$\half\covered{C}^*_{0} $ (equivalently, from the half-plane $\Re z>0$)
to their whole domain $\covered{C}^* $.

At the same time, \Eqs{~}\eqref{eq::0860}, \eqref{eq::0870} may be regarded as yet
another set  of 
constraints which the  eigenfunctions $\EEpm{\pm}  $ always obey.

Proceeding with applications of the above results,
let us now consider the following non-linear first order differential equation
\begin{equation}
  \label{eq::0880}
  \dot\varphi(t)+\sin\varphi(t)=B+A\cos\omega{}t,
\end{equation}
in which $A,B,\omega$ are real constants, $A\not=0,\omega>0$, 
and dot denotes  derivative with respect to the free real variable $t$.
This equation
and its generalizations
appear 
in a number of problems of physics
(most notably, in the modeling of
Josephson junctions 
\cite{MK, St, McC}),
mechanics
\cite{Foo,LT,FLT}, 
dynamical systems theory
\cite{  GI},
geometry
\cite{BLPT}.
On the other hand, Eq.~\eqref{eq::0880}   is closely linked 
to the double confluent Heun equation \cite{T3}. 
More exactly, it is just the family of equations
of the form \eqref{eq::0010} which are equivalent to Eq.~\eqref{eq::0880}. 
The transformation connecting the triplets of their parameters
reads
\begin{equation}
B =\omega l,
\quad
A = 2\omega \mu,\quad,(2\omega)^{-1}=\sqrt{\lambda+\mu^2}.
\end{equation}
The last equation agrees with condition \eqref{eq::0530} 
but, to keep all the parameters real, one has to
impose the restriction $\lambda+\mu^2>0$, more severe than  \eqref{eq::0400}.
Moreover, 
$\lambda$ and $\mu$ are to be assumed to be real themselves as well.

The transition from Eq.~\eqref{eq::0010} to Eq.~\eqref{eq::0880}
is carried out in two steps \cite{T4}.
First, we replace Eq.~\eqref{eq::0880} with the following Riccati
equation 
\begin{equation}
          \label{eq::0900}
{\eIphi}'
+(2\Imi \omega)^{-1}z^{-1}({\eIphi}^2-1)=
\big(lz^{-1} 
+\mu(1+z^{-2})\big){\eIphi}
\end{equation}
for holomorphic function $\eIphi=\eIphi(\coveredz)$.
$\eIphi$
is introduced as
an analytic continuation
of the real analytic function $e^{\Imi\varphi(t)}$ from 
(the lift of)
the unit circle, i.e.\ it 
obeys the equation
\begin{equation}
\label{eq::0910}
\eIphi((\iota^{-1})e^{\Imi\omega t})
=
e^{\Imi\varphi(t)},
\;
t\in(-\pi/\omega,\pi/\omega),
\end{equation}
where $ (\iota^{-1}) $ denote the lift of $\mathbb{C}^*\fgebackslash\mathbb{R}_-  $
to $\covered{C}^*$ such that $(\iota^{-1})1=\coveredone$.
We shall need also the function $\eP=\eP(\coveredz)$ which is a similar 
analytic continuation 
of the exponentiated quadrature 
$\int_0^t\cos\varphi(\tilde t)\,d\tilde t=P(t) $, i.e., obeys
the equation
\begin{equation}
\label{eq::0920}
\eP\big((\iota^{-1})e^{\Imi\omega t}\big)
=e^{P(t)}
=
\exp{   \int_0^t\!\!\cos\varphi(\tilde t)\,d\tilde t  },
\;
t\in(-\pi/\omega,\pi/\omega),
\end{equation}
the corresponding differential equation reading
\begin{equation}
\label{eq::0930}
2\Imi\omega \eP'=z^{-1}( {\eIphi} + {\eIphi}^{-1})\eP.
\end{equation}
Obviously, at the point $\coveredz=\coveredone$, it holds
\begin{equation}
\label{eq::0940}
|\eIphi(\coveredone)|=1,\mbox{ }\eP(\coveredone)=1.
\end{equation}

Eq.\ \eqref{eq::0880} is equivalent to Eq.\ \eqref{eq::0900} 
considered in vicinity of (the lift of) the unit circle with $-1$ removed
for which the initial condition posed at $\coveredz=\coveredone$
obeys the constraint  \eqref{eq::0940}.
This relation
is ensured by explicit
locally invertible transformations of the 
independent variables and unknown functions.

In the second step, the relations between Eq.{~}\eqref{eq::0010} and
Eq.{~}\eqref{eq::0900} are established in the form
transformations taking solutions to one of them to solutions to the
other equation. More exactly, 
as long as one concerns
 Eq.~\eqref{eq::0010}, it proves more convenient to consider
not arbitrary solutions but the eigenfunctions $\EEpm{\pm}$
of the operator $\opmC$.

In this way,
the transformation
taking the functions $\EEpm{\pm}$ to solutions to Eq.~\eqref{eq::0010}
is as follows
\begin{equation}
\label{eq::0950}
\eIphi(\coveredz)=
-\Imi  z^{-\mllll}
\frac{
\cos(\half\alpha) \EEpm{+}(\coveredz)
+\Imi
\sin(\half\alpha) \EEpm{-}(\coveredz)
}{
\cos(\half\alpha) \EEpm{+}(1/\coveredz)
-\Imi
\sin(\half\alpha) \EEpm{-}(1/\coveredz)
}.
\end{equation}
Here $\alpha$ is an arbitrary real constant. Since 
in accordance with  \eqref{eq::0910}
$$
\eIphi(\coveredone)=e^{\Imi\varphi(0)}=
-\Imi{}
\frac{
\cos(\half\alpha) \EEpm{+}(\coveredone)
+\Imi
\sin(\half\alpha) \EEpm{-}(\coveredone)
}{
\cos(\half\alpha) \EEpm{+}(\coveredone)
-\Imi
\sin(\half\alpha) \EEpm{-}(\coveredone)
},
$$
where both $ \EEpm{\pm}(\coveredone)$ are non-zero,
it is obvious that,
varying $\alpha$ through the interval $[0,2\pi)$ (i.e., though a circle/projective line $P(\mathbb R)$),
one is able to obtain all solutions to Eq.\ \eqref{eq::0880} 
(up to uncertainty allowing additions of integer multiples of $2\pi$).
Let us give 
also 
similar relation 
yielding the correspondingly transformed $\eP$:
\begin{eqnarray}
\label{eq::0960}
\begin{aligned}
\eP(\coveredz)=&
\big(
\cos^2(\half\alpha)
\EEpm{+}^2(\coveredone)
+
\sin^2(\half\alpha)
\EEpm{-}^2(\coveredone)
  \big)^{-1}\times
\\&\;
e^{\mu(2-z-1/z)}
\big({
\cos(\half\alpha) \EEpm{+}(\coveredz)
+\Imi
\sin(\half\alpha) \EEpm{-}(\coveredz)
}\big)\times
\\&\;\phantom{e^{\mu(2-z-1/z)}  }
\big({
\cos(\half\alpha) \EEpm{+}(1/\coveredz)
-\Imi
\sin(\half\alpha) \EEpm{-}(1/\coveredz)
}\big).
\end{aligned}
\end{eqnarray}
Notice that one has $\eP(\coveredone)=1$ for it, the first factor
just ensuring such a normalization.
Besides, varying independently the normalizations of
the functions
$\EEpm{\pm}$, i.e.,
for instance, 
the values of
$\EEpm{\pm}(\coveredone)\not=0$, the solution $\eIphi$ and 
the function
$\eP$ vary coherently,
this time 
in a non-trivial way.

The validity of the above transformations is verified by means of 
the plugging their right-hand sides into the corresponding equations,
and elimination of the derivatives $\EEpm{\pm}'$ with the help of 
\Eqs{~}\eqref{eq::0610}.

The inverse transformations,
taking an arbitrary solution (up to a minor exception, see below)
to Eq.\ \eqref{eq::0880}
to the eigenfunctions 
$\EEpm{\pm}$ of the operator $\opmC$,
fulfilling Eq.\ \eqref{eq::0010} as well,
read 
%
\begin{equation}
\label{eq::0970}
\begin{aligned}
\EEpm{\pm}(\coveredz)&=
\half 
e^{\mu(z+1/z-2)/2}
z^{{\mllll}/2}
\left\lgroup
\frac{1\pm\Imi}{\sqrt2}
\eP(\coveredz)^{1/2}\eIphi(\coveredz)^{1/2}
+
\frac{1\mp\Imi}{\sqrt2}
\eP(1/\coveredz)^{1/2}\eIphi(1/\coveredz)^{-1/2}
\right\rgroup.
\end{aligned}
\end{equation}
The fulfillment of 
\Eqs{~}\eqref{eq::0610} is verified by straightforward 
substitution and application of \Eqs{~}\eqref{eq::0900}, \eqref{eq::0930}.

Taking into account Eq.\ \eqref{eq::0910}, one obtains
the equation
\begin{equation}
\label{eq::0980}
\EEpm{\pm}(\coveredone)
=
\mp\eP(\coveredone)\sin \half(\varphi(0)\pm\pi/2),
\end{equation}
manifesting 
the specific normalization of $\EEpm{+}$ and $\EEpm{-}$
innate to the representations \eqref{eq::0970}.
It reveals a peculiarity of the solutions $\varphi(t)$ to Eq.\ \eqref{eq::0010} 
solving one of initial data problems with
$$
\varphi(0)=\pi/2(\hspace{-1em}\mod\pi)=\pm\pi/2(\hspace{-1em}\mod2\pi), \mbox{ implying } \eIphi(\coveredone)=\pm\Imi.
$$
For such a $\eIphi$ the pair of formulas \eqref{eq::0970} yields only one
non-trivial eigenfunction because instead of 
the second one the identically zero function is produced. 
It also means that in such cases
for one of two choices 
of the sings
the expression brackets in \eqref{eq::0970} vanishes, i.e.\ it holds
$$
\eP(\coveredz)\eIphi(\coveredz)
=
\eP(1/\coveredz)/\eIphi(1/\coveredz),
$$

The close relation of Eq.\ \eqref{eq::0010} and Eq.\ \eqref{eq::0900}
outlined above
assumes allusion to 
possible manifestation of 
symmetries the space of solutions to the former possesses, if any,
in properties of the set of solutions to the latter.
Specifically, we have noted above a number of maps 
sending solutions to Eq.\ \eqref{eq::0010} to solutions
of the same equation. One may expect that there exist 
associated maps acting this time to solutions to \eqref{eq::0900}
with similar composition properties.
Below we analyze such an opportunity in more details
for operators given above in explicit form.

To begin,
let us note first that the numerator of the 
fraction in the right-hand side of \eqref{eq::0950}
is a solution to Eq.\ \eqref{eq::0010}.
The denominator is also a solution but with
the argument inversed. 
If one applies
a linear operator taking them to other solutions
then the latter 
can also be used as numerators and
denominators in the formula
\eqref{eq::0950} and such a transformation has finally effect of
certain modification of the parameter $\alpha$. 
Replacing finally therein the functions $\EEpm{\pm}$
by their representations
in terms of $\eIphi$ and $\eP$ in accordance with 
Eq.\ \eqref{eq::0970}, the closed transformation 
acting to solutions to \eqref{eq::0900}
and referring to nothing else
should result.

A toy example is provided by consideration of the operator $\opmC$.
Specifically, let us consider \eqref{eq::0950} with $\alpha=\pi/4$
and plug in it the expressions of $\EEpm{\pm}$ in terms of $\eIphi$
and $\eP$ given in \Eqs{~}\eqref{eq::0970}. The result is the original
solution $\eIphi(\coveredz)$. However, if we apply on a course of transformation
the  operator $\opmC$ to its eigenfunctions $\EEpm{\pm}$, i.e.\ keep
$\EEpm{+}$ unchanged but replace $\EEpm{-}$ with $-\EEpm{-}$, then the result becomes 
equal to 
$ -\eIphi(1/\coveredz)^{-1}$. This function verify Eq.\ \eqref{eq::0900}
as long as $\eIphi(\coveredz)$ does. The transformation
\begin{equation}
\label{eq::0990}
\eIphi(\coveredz)\to -\eIphi(1/\coveredz)^{-1}
\end{equation}
is therefore the map of the set of solutions to  Eq.\ \eqref{eq::0900}
onto itself. In accordance with way of its derivation, it is also the 
representation of the operator $\opmC$ on the latter set. 
It is also worth noting that, when restricted to the (lift of) unit circle,
Eq. \eqref{eq::0990} converts to a more or less evident transformation
\begin{equation}
\label{eq::1000}
\varphi(t)\to\pi-\varphi(-t)
\end{equation}
leaving Eq.\ \eqref{eq::0880} invariant.

Proceeding next with the less trivial \monodromy{} transformation,
it is useful
 to represent 
in advance
its matrix  \eqref{eq::0810}
(which serves for transformation of  solutions to Eq.~\eqref{eq::0010}, not to Eq.~\eqref{eq::0900})
in terms solutions to Eq.~\eqref{eq::0900}. 
In such a representation it is equal to
\begin{eqnarray} 
\label{eq::1020}
&\begin{aligned} 
\mathbf{M}=&
{(-1)^\mllll}{ e^{-P(0)}\sec\varphi(0) }
\begin{pmatrix}
\beta_+ 
&
\Imi(\beta_-+\gamma) %
\\
\Imi(-\beta_-+\gamma)
&
\beta_+
\end{pmatrix},
\end{aligned}
\\[-3ex]
\nonumber
&\begin{aligned}
\\\mbox{where}\;
\beta_{\pm}=&\:\half\big(e^{P(T/2)}\cos\varphi(\half T)\pm e^{P(-T/2)}\cos\varphi(-\half T)\big),
\mbox{ where }T=2\pi/\omega,
\\
\gamma=&\:e^{P(T/2)/2+P(-T/2)/2}\sin\big(\half \varphi(\half T)-\half\varphi(-\half T)\big).
\end{aligned}
\end{eqnarray}
Utilizing this result, 
one can obtain the following explicit representation 
of the coupled \monodromy{} transformations of solutions to
\Eqs{~}\eqref{eq::0900} and \eqref{eq::0930}. They read
\begin{equation}
\label{eq::1030}
\begin{aligned}
\Mono\eIphi(\coveredz)=&\;
\Bigl(\hphantom{\pm}
   e^{P(T/2)/2}{}\cos\bigl(\varphi(\half{}T)\bigr)
{\cdot}{\eP(\coveredz)}^{1/2}  {}{\eIphi(\coveredz)}^{1/2}
\\&
+\Imi{\,}e^{P(-T/2)/2}{}\sin\bigl((\varphi(\half{}T)-\varphi(-\half{}T))/2\bigr)
{\cdot}{\eP(1/\coveredz)}^{1/2} {\eIphi(1/\coveredz)}^{-1/2}
\Bigr)\times
{}
\\&
\Bigl(\hphantom{\pm}
   e^{P(T/2)/2}{}\cos\bigl(\varphi(\half{}T)\bigr)
{}{\eP(\coveredz)}^{1/2}   {\eIphi(\coveredz)}^{-1/2}
\\&
   -\Imi{\,}e^{P(-T/2)/2}{}\sin\bigl((\varphi(\half{}T)-\varphi(-\half{}T))/2\bigr)
{}{\eP(1/\coveredz)}^{1/2}{}{\eIphi(1/\coveredz)}^{1/2}
\Bigr)^{-1},
\end{aligned}
\end{equation}
\begin{equation}
\label{eq::1040}
\begin{aligned}
\Mono\eP(\coveredz)=&\;
  2{}e^{-P(-T/2)}{}\cos\bigl(\varphi(0)\bigr)
 \sec^2\bigl((\varphi(\half{}T) + \varphi(-\half{}T))/2\bigr)\times
  \\&
   {}\Bigl(\hphantom{\pm}
     e^{P(T/2)}{}\cos^2\bigl(\varphi(\half{}T)\bigr)
        {\cdot}{\eP(\coveredz)}^{1/2}{\eP(1/\coveredz)}^{-1/2}+
\\&
        + e^{P(-T/2)}{}\sin^2\bigl((\varphi(\half{}T) - \varphi(-\half{}T))/2\bigr)\
        {\cdot}{\eP(1/\coveredz)}^{1/2}{\eP(\coveredz)}^{-1/2}
        \\&
    -\Imi{\,}\cos\bigl(\varphi(T/2)\bigr)
       {}e^{(P(T/2) + P(-T/2))/2}
        {}\sin\bigl((\varphi(\half{}T) - \varphi(-\half{}T))/2\bigr)\cdot
        \\&\hphantom{ -\Imi{\,}\cos\bigl(\varphi(T/2)\bigr) }
        {\cdot}(
          {\eIphi(1/\coveredz)}^{1/2}{}{\eIphi(\coveredz)}^{1/2} 
        - {\eIphi(1/\coveredz)}^{-1/2}{}{\eIphi(\coveredz)}^{-1/2}
          )
     \Bigr)\times
     \\&
     \bigl(
           {\eIphi(\coveredz)}^{1/2}{}{\eIphi(1/\coveredz)}^{1/2}
            + {\eIphi(\coveredz)}^{-1/2}{}{\eIphi(1/\coveredz)}^{-1/2}  
     \bigr)^{-1},
\end{aligned}
\end{equation}
Their restriction to the lift of the unit circle
leads to the 
transformations of solutions 
to Eq.~\eqref{eq::0880}
\begin{equation}
\label{eq::1050}
 \begin{aligned}
e^{\Imi\varphi(t \pm T)}=
&\big(\hphantom{\pm\Imi}
 e^{ P(\pm T/2)/2}\cos(\varphi(\pm \half{}T)/2)
\cdot e^{P(t)/2}e^{\Imi\varphi(t)/2}
\\&
\pm\Imi
e^{P(\mp \half{}T)/2}\sin((\varphi(\half{}T)-\varphi(-\half{}T))/2)
\cdot e^{P(-t)/2}e^{-\Imi\varphi(-t)/2}
\big)\times
\\
&\big(\hphantom{\pm\Imi}
 e^{ P(\pm \half{}T)/2}\cos(\varphi(\pm \half{}T)/2)
\cdot e^{P(t)/2}e^{-\Imi\varphi(t)/2}
\\&
\mp\Imi
e^{ P(\mp T/2)/2}\sin((\varphi(\half{}T)-\varphi(-\half{}T))/2)
\cdot e^{P(-t)/2}e^{\Imi\varphi(-t)/2}
\big)^{-1},
\end{aligned}
\end{equation}
 \begin{equation}
\label{eq::1060}
 \begin{aligned}
e^{P(t \pm T)}=
&\cos(\varphi(0))\,
e^{-P(\mp T/2)}
\sec^2\!\big(( \varphi(\half{}T)+\varphi(-\half{}T))/2\big)
\times
\\&
\big\{
\hspace{0.81em}
e^{P(\pm T/2)}
\cos^2\!\big(\varphi(\pm\half{}T)/2\big)
\cdot e^{(P(t)-P(-t))/2}
\\&
+
e^{P(\mp T/2)}\sin^2\!\big(( \varphi(\half{}T)-\varphi(-\half{}T))/2\big)
\cdot e^{(P(-t)-P(t))/2}
\\&
\pm
2
e^{(P(T/2)+P(-T/2))/2}
\sin\big(( \varphi(\half{}T)-\varphi(-\half{}T))/2\big)
\cos\big(\varphi(\pm\half{}T)/2\big)
\cdot
\\&
\hspace{2em}\sin\big(( \varphi(t)+\varphi(-t))/2\big)
\big\}
\cdot
 \sec\big(( \varphi(t)+\varphi(-t))/2\big)
\end{aligned}
\end{equation}
which extend them 
from the segment of variation
 of $t$ of the length equal to the period $T$ of the right-hand side of
Eq.~\eqref{eq::0880}
to the both adjacent segments of the same length.

It is worth noting that the usual interpretation of \monodromy{}
refers to the behavior of solutions of a differential
equation along loops encircling the chosen singular point of
this equation (if there are several ones). These loops
can be arbitrarily small and even may contract to this point,
the \monodromy{} being insensitive to their continuous deformations.
In particular, such approach applies in case of the linear equation \eqref{eq::0010}
for which the only singular point (in $\mathbb C$) is zero.
Its solutions are regular everywhere (with a natural exception) and the loop
can be
arbitrarily small as well as
arbitrarily large.
However,
in case of the functions $\eIphi$ and $\eP$
of which the former
obeys the non-linear equation \eqref{eq::0900},
the \monodromy{} described by the formulas \eqref{eq::1030}, \eqref{eq::1040}
is determined, basically, along the loop coinciding with the unit circle.
It also can be deformed but not contracted to zero since the unit
disk always contain the singular points of the function $\eIphi $. In accordance
with \eqref{eq::0950} they are
associated
with
roots of
certain linear combinations of the functions $\EEpm{\pm}$. Since the latter verify
a linear homogeneous ODE, these roots are simple and hence all the singularities of  $\eIphi$,
except at zero,
are simple poles.

\Eqs{~}\eqref{eq::1040} and \eqref{eq::1050}
do not involve the parameter $\mllll=-l$ which had been restricted here to positive integers.
Hence it can be supposed that such a limitation is superfluous for their validity.
It is indeed the case and a straightforward computation shows
that the above  transformations  hold true as long as $\varphi$ is a solution 
to Eq.~\eqref{eq::0880} irrespectively to the values of the constant parameters.

On the contrary, the transformation of solutions to Eq.~\eqref{eq::0900}
generated by the operators $\opmA$ and $\opmB$
can be constructed only if $\mllll\in\mathbb{N}$.
They can be represented in the following way.

\newcommand{\AAACCC}{u}
\newcommand{\BBBDDD}{v}
\newcommand{\KKKLLL}{w}

\newcommand{\AAA}{\AAACCC_{\!+}}
\newcommand{\CCC}{\AAACCC_{\!-}}
\newcommand{\BBB}{\BBBDDD_{\!+}}
\newcommand{\DDD}{\BBBDDD_{\!-}}
\newcommand{\KKK}{\KKKLLL_{\!+}}
\newcommand{\LLL}{\KKKLLL_{\!-}}

\begin{eqnarray}          \label{eq::1070}&&
\begin{aligned}
\eIphi_A(\coveredz)=
\frac{
V \eP^{1/2}(\coveredz)\eIphi^{1/2}(\coveredz) - \Imi\, U \eP^{1/2}(1/\coveredz)\eIphi^{-1/2}(1/\coveredz)
}{
V \eP^{1/2}(\coveredz)\eIphi^{-1/2}(\coveredz) + \Imi\, U \eP^{1/2}(1/\coveredz)\eIphi^{1/2}(1/\coveredz)
},
\end{aligned}
 \\                        \label{eq::1080}&&
\begin{aligned}
\eP_A(\coveredz)=&\;
e^{-P(0)}(U^2+V^2-2U V\sin\varphi(0))^{-1}\times
\\&\hphantom{\;e^{-P(0)}   }
(
V \eP^{1/2}(\coveredz)\eIphi^{1/2}(\coveredz) - \Imi\, U \eP^{1/2}(1/\coveredz)\eIphi^{-1/2}(1/\coveredz)
)\times
\\&\hphantom{\;e^{-P(0)}   }
(
V \eP^{1/2}(\coveredz)\eIphi^{-1/2}(\coveredz) + \Imi\, U \eP^{1/2}(1/\coveredz)\eIphi^{1/2}(1/\coveredz)
),
\end{aligned}
\\                          \label{eq::1090}&&
\begin{aligned}
\eIphi_B(\coveredz)=
\frac{
U \eP^{1/2}(\coveredz)\eIphi^{1/2}(\coveredz) - \Imi\, V \eP^{1/2}(1/\coveredz)\eIphi^{-1/2}(1/\coveredz)
}{
U \eP^{1/2}(\coveredz)\eIphi^{-1/2}(\coveredz) + \Imi\, V \eP^{1/2}(1/\coveredz)\eIphi^{1/2}(1/\coveredz)
},
\end{aligned}
\\                          \label{eq::1100}&&
\begin{aligned}
\eP_B(\coveredz)=&\;
e^{-P(0)}(U^2+V^2-2U V\sin\varphi(0))^{-1}\times
\\&\hphantom{\;e^{-P(0)}   }
(
U \eP^{1/2}(\coveredz)\eIphi^{1/2}(\coveredz) - \Imi\, V \eP^{1/2}(1/\coveredz)\eIphi^{-1/2}(1/\coveredz)
)\times
\\&\hphantom{\;e^{-P(0)}   }
(
U \eP^{1/2}(\coveredz)\eIphi^{-1/2}(\coveredz) + \Imi\, V \eP^{1/2}(1/\coveredz)\eIphi^{1/2}(1/\coveredz)
),
\end{aligned}
\\[-2em]\nonumber
\end{eqnarray}
\begin{eqnarray}
\mbox{where }
\AAACCC_{\pm} &= 
&
 e^{\Imi\varphi(T/2)/2}
 \pm(-1)^{\mllll}
 \Imi\, e^{-\Imi\varphi(T/2)/2},
 \nonumber
\\
 \BBBDDD_{\pm} & =
 &
 (-1)^{\mllll}
 e^{\Imi\varphi(-T/2)/2}
 \pm
 \Imi\,e^{-\Imi\varphi(-T/2)/2},
 \nonumber
\\
 \KKKLLL_{\pm}&=
 &
  \cos(\half\varphi(0)) \pm \sin(\half\varphi(0)),
\\
U&=&
2e^{P(T/2)/2}
\big(
\delTa_{+}\AAACCC_{+} \KKKLLL_{+}
-\Imi\,\delTa_{-}\AAACCC_{-} \KKKLLL_{-}
\big),
\nonumber
\\
V&=&\hphantom{ +\Imi\, }
\delTa_{+}\KKKLLL_{+}
\big(e^{P(T/2)/2}\AAACCC_{+}
+
e^{P(-T/2)/2}\KKKLLL_{+}
\big)
\nonumber
\\&&
+\Imi\,
\delTa_{-}\KKKLLL_{-}
\big(e^{P(T/2)/2}\AAACCC_{-}
+
e^{P(-T/2)/2}\KKKLLL_{-}
\big).
\nonumber
\end{eqnarray}
In a sense, these transformations inherit the composition properties of 
their prototypes. The corresponding 
results can be resumed as follows

\begin{theorem}
Let a solution $\eIphi $
to Eq.{~}\eqref{eq::0900}
holomorphic
in
some vicinity of the lift of
$\puncturedS_{(-1)}$ be given and 
$\eP=\eP(z) $ be a  solution to Eq.{~}\eqref{eq::0930},
the conditions \eqref{eq::0940} being fulfilled.
Then
\begin{itemize}
\item
\Eqs{~}\eqref{eq::1030}, \eqref{eq::1040},
in which $\varphi(0), \varphi(\pm\half T), P(0), P(\pm\half T) $
are determined from \Eqs{~}\eqref{eq::0910} and \eqref{eq::0920}, 
represent the \monodromy{} transformations of
projections of 
$\eIphi$
and $\eP$ for the paths homotopic to $S^1$.
\item
The transformations \eqref{eq::1070}-\eqref{eq::1100} take these solutions 
to solutions of the same equations. Moreover,
when applied twice, 
the maps
\eqref{eq::1070}, \eqref{eq::1080} reproduce the original $\eIphi,\eP  $
while two-fold maps
\eqref{eq::1090}, \eqref{eq::1100}
become the above \monodromy{} map.
\end{itemize}
\end{theorem}
Thus the transformation 
represented by \Eqs{~}%
\eqref{eq::1070}, \eqref{eq::1080} is \involutive{}
and transformation 
represented by \Eqs{~}%
\eqref{eq::1090}, \eqref{eq::1100}
can be called a square root of the \monodromy{} transformation.

The final remark concerning  implications of symmetries of the space of solutions to 
Eq.{~}\eqref{eq::0010} for solutions to Eq.{~}\eqref{eq::0880} is as follows.
The \monodromy{} transformation
in the form
\eqref{eq::0810}
enabled us to realize  the algebraic extension of solutions to 
Eq.{~}\eqref{eq::0010}
from the \subdomain{} with closure of projection coinciding with the complex plane
to their whole domains.
In terms of solutions to Eq.{~}\eqref{eq::0880} such extension looks like 
their continuation from the segment $(-T/2,T/2)$ of length $T=2\pi/\omega$
to the whole real axis. 
There are also transformations Eq.{~}\eqref{eq::0840}-\eqref{eq::0870}
which result in possibility of `dissemination' of solutions to Eq.{~}\eqref{eq::0010}
from the half-plane $\Re z>0 $.
When restricted to the unit circle, this relation means possibility
of algebraic extending of solutions to Eq.{~}\eqref{eq::0880}
from the segment $(-T/4,T/4)$ of length $T/2$.

\bigskip
\noindent{\bf Appendix}\medskip

\noindent
When considered 
on simply connected subsets of $\mathbb{C}^*$,
Eq.{~}\eqref{eq::0010} determines solutions
which are single-valued holomorphic functions.
In particular, this holds
for projections (in fact, restrictions)
$\iota_*\EEpm{+}(z), \iota_*\EEpm{-}(z)$
of the eigenfunctions $\EEpm{+}(\coveredz), \EEpm{-}(\coveredz)$ of the operator $\opmC $
to
$\mathbb{C}^{-*}= \mathbb{C}^{*} \fgebackslash\mathbb{R}_-$,
the
complex plane without zero
with negative real axis $\mathbb{R}_-$ removed.
Their one-side limits $\iota_*\EEpm{\pm}(x+0\cdot\Imi)$ and
$\iota_*\EEpm{\pm}(x-0\cdot\Imi)$
at points $x$ of $\mathbb{R}_-$ 
 approached from above ($x+0\cdot\Imi$)
and from below ($x-0\cdot\Imi$), respectively,
do not coincide but are connected by the linear transformation
determined by the matrix \eqref{eq::0810} (by definition, the monodromy transformation)
{\em which does not depend on $x$.}
Notice also that if
\begin{equation}
\label{eq::1120}
\EEpm{+}(\coveredrevco\!\coveredone)^2-\EEpm{+}(\coveredrevpro\!\coveredone)^2
\not=0\not=
\EEpm{-}(\coveredrevco\!\coveredone)^2-\EEpm{-}(\coveredrevpro\!\coveredone)^2
\end{equation}
then the non-diagonal   elements of \eqref{eq::0810} are non-zero and
each function
$\EEpm{+}$ or
$\EEpm{-}$, as well as their projections $\iota_*\EEpm{\pm}$,
is taken 
to
a function linearly independent of it.
At the same time there exist solutions to Eq.{~}\eqref{eq::0010} for which
the same transformation does not engender a distinct solution but
reduces
to multiplication of the argument by some constant alone.
Indeed, let us define the two particular linear  combinations
$E_{[+]}^{(M)}$ and $E_{[-]}^{(M)}$
of 
$\EEpm{+}$ and 
$\EEpm{-}$ as follows:
\begin{equation}
\label{eq::1130}
E_{[\pm]}^{(M)}
=
\big(\Imi(\EEpm{-}(\coveredrevco\!\coveredone)^2-\EEpm{-}(\coveredrevpro\!\coveredone)^2)\big)^{1/2}
\EEpm{+}
\pm
\big(\Imi(\EEpm{+}(\coveredrevco\!\coveredone)^2-\EEpm{+}(\coveredrevpro\!\coveredone)^2)\big)^{1/2}
\EEpm{-}.
\end{equation}
A straightforward computation using \Eqs{~}\eqref{eq::0830},\eqref{eq::0810},
which describe the monodromy transformation $\Mono$ of $\EEpm{\pm}$ in explicit form,
shows that
\begin{equation}
\label{eq::1140}
\Mono E_{[\pm]}^{(M)}=
\Lambda_{\pm}E_{[\pm]}^{(M)},
\end{equation}
where the complex (or, perhaps, real) numbers
\begin{equation}
\begin{aligned}
\label{eq::1150}
\Lambda_{\pm}=&e^{4\mu}\big(2\EEpm{+}(\coveredone)\EEpm{-}(\coveredone)\big)^{-1/2}
\Bigl(\EEpm{+}(\coveredrevco\!\coveredone)\EEpm{-}(\coveredrevco\!\coveredone)
+
\EEpm{+}(\coveredrevpro\!\coveredone)\EEpm{-}(\coveredrevpro\!\coveredone)
\\&
\pm
\big(
(\EEpm{+}(\coveredrevco\!\coveredone)^2-\EEpm{+}(\coveredrevpro\!\coveredone)^2)
(\EEpm{-}(\coveredrevco\!\coveredone)^2-\EEpm{-}(\coveredrevpro\!\coveredone)^2)
)^{1/2}
\Bigr)
\end{aligned}
\end{equation}
are the eigenvalues of the matrix \eqref{eq::0810} and, at the same time, the eigenvalues
of the monodromy operator itself. It is worth noting that
\begin{equation}
\label{eq::1160}
\Lambda_{+}\cdot\Lambda_{-}=1
\end{equation}
(this follows from definitions \eqref{eq::1150} and  Eq.~\eqref{eq::0570}; yet
another way of proving leans on the first equation among \Eqs{~}\eqref{eq::0820}).
The inequalities \eqref{eq::1120} are precisely the condition of non-coincidence of these
eigenvalues. If it is fulfilled, the functions $E_{[\pm]}^{(M)}$
are linear independent and
constitute the basis of the space
of solutions to 
Eq.{~}\eqref{eq::0010}.

Eq.{~}\eqref{eq::1140} states in particular that
on the edges  of $\mathbb{C}^{-*}$
contacting
$\mathbb{R}_-$ the projections of $E_{[\pm]}^{(M)}$
obey the constraints 
\begin{equation}
\label{eq::1170}
\iota_*E_{[\pm]}^{(M)} (x+0\cdot\Imi)=
\Lambda_{\pm}\cdot\iota_*E_{[\pm]}^{(M)} (x-0\cdot\Imi),\; x\in \mathbb{R}_-.
\end{equation}
The power functions $z^{(2\Imi\pi)^{-1}{\mathrm{Log}}\Lambda_{\pm}}   $
possess exactly the same property.
This means that for each choice of the signs
the one-side limiting values
of the product $z^{\Imi(2\pi)^{-1}{\mathrm{Log}}\Lambda_{\pm}}\cdot\iota_*E_{[\pm]}^{(M)} (z) $
on the edges of $\mathbb{C}^{-*}$
{\em coincide\/}.
Accordingly,
they provide its  extension  to $ \mathbb{R}_- $
yielding 
a  single-valued continuous function $G_{\{\pm\}}(z)$.
In view of
existence and
uniqueness of
analytic continuations across  $ \mathbb{R}_- $ in both directions, this function is
holomorphic everywhere in $\mathbb{C}^{*}$. We have established, therefore, the following result.
                       \begin{theorem} 
Let $\EEpm{+}(\coveredz)$ and $\EEpm{-}(\coveredz)$ be solutions to Eq.{~}\eqref{eq::0010}
such that $\EEpm{\pm}(\coveredone)\not=0$,
and
the derivatives
$\EEpm{\pm}'(\coveredone)$ obey \Eqs{~}\eqref{eq::0510}.
Let the constant parameters of Eq.{~}\eqref{eq::0010} be such that
the conditions \eqref{eq::1120} are fulfilled.
Let us define the constants
 $\gamma_\pm=\Imi(2\pi)^{-1}{\mathrm{Log}}\Lambda_{\pm}$,
where the numbers $ \Lambda_{\pm} $ defined by the formulas \eqref{eq::1150}
are non-zero in view of Eq.~\eqref{eq::1160} which also implies that $  \gamma_++\gamma_-=0$.
Then the
functions $ \coveredz^{\gamma_\pm}E_{[\pm]}^{(M)} (\coveredz) $,
where
$ \coveredz^{\gamma_\pm} $ are the lifts
 of the power functions
$z^{\gamma_\pm}$ such that $ \coveredone^{\gamma_\pm}=1 $,
coincide with 
the lifts (pullbacks) to $ \covered{C}^* $ of some functions $ G_{\pm}(z) $ holomorphic on
$\mathbb{C}^*$, i.e.\
\begin{equation}
\label{eq::1180}
\coveredz^{\gamma_\pm}E_{[\pm]}^{(M)} (\coveredz)
=G_{\pm}(\iota\, \coveredz).
\end{equation}
                      \end{theorem}   
\begin{corollary}
Under conditions assumed by the above theorem,
any solution to Eq.{~}\eqref{eq::0010}
is a linear combination of two products of a power function and a function 
holomorphic on $\mathbb{C}^*$.
\end{corollary}
\noindent
In accordance with Laurent theorem, the functions $ G_{\pm}(z) $
admit the expansions in Laurent series with center at zero which converge
everywhere except at zero. In terms of properties of solutions to
Eq.{~}\eqref{eq::0010} this means the following.
\begin{corollary}
Under conditions assumed by the above theorem,
there exists $\gamma\not=0$ such that
there exist the two two-sides sequences $g_{+}^{(k)}$ and $g_{-}^{(k)}$, $k\in\mathbb{Z}$,
obeying the equations
\begin{equation}
\label{eq::1190}
0=
-\mu(k \pm\gamma+l)g_{\pm}^{(k-1)}
+\big( (k \pm\gamma) (k \pm\gamma+l) +\lambda \big)g_{\pm}^{(k)}
+\mu(k \pm\gamma+1)g_{\pm}^{(k+1)}
\end{equation}
such that the two generalized power series
\begin{equation}
\label{eq::1200}
\somE_{[\pm]}(z)=\sum_{k=-\infty}^\infty g_{\pm}^{(k)} z^{k\pm\gamma}
\end{equation}
converge everywhere except at zero
and satisfy Eq.{~}\eqref{eq::0010}. \\
These are permuted by
the operator $\opmC$ (up to some numerical factors).
         \end{corollary}
The use of power series with center at zero as a ``template'' for solutions to
DCHE including Eq.{~}\eqref{eq::0010} is the wellknown method, see, for example,
Eq.~(2.4.54), the next one, and Eq.~(2.4.55) in Ref.~\cite{SW}.
Its curious feature is that, solving \Eqs{~}\eqref{eq::1190},
the series \eqref{eq::1200} can be constructed for {\em arbitrary\/} $\gamma$.
Plugging it further into Eq.{~}\eqref{eq::0010} and carrying out computations
term by term, one finds that the latter is satisfied. The point however is that such a solution
is actually {\em formal}, i.e.\
the series
\eqref{eq::1200} diverges for any non-zero $z$.
The only exception (up to some formally distinct but equivalent ones)
corresponding to an actual solutions to Eq.{~}\eqref{eq::0010}
is the choice specified by the above theorem. The reasoning leading to it
provide us with a simple {\em proof on existence of $\gamma$
making some solution to \Eqs{~}\eqref{eq::1190} convergent}.
Another proof based on the Hadamard-Perron theorem had been
proposed in Ref.\ \cite{BG}.

\end{document}